\documentclass[final]{IEEEtran}
\usepackage[utf8]{inputenc}
\usepackage{bbm}
\usepackage{graphicx}
\usepackage{placeins}
\usepackage{color,soul}
\usepackage{amsthm}
\usepackage{amsmath}
\usepackage{mathtools}
\usepackage{amsfonts}
\def\Pr{\mathbb{P}}
\def\peq#1{\stackrel{\text{\scriptsize(#1)}}{=}}
\def\pleq#1{\stackrel{\text{\scriptsize(#1)}}{\leq}}
\def\pgeq#1{\stackrel{\text{\scriptsize(#1)}}{\geq}}
\def\Ex{\mathbb{E}}
\def\ind{\mathbbm{1}}
\def\cal{\mathcal}
\def\dd{\mathrm{d}}

\DeclareMathOperator*{\argmax}{arg\,max}

\mathchardef\mhyphen="2D
\usepackage{newtxtext,newtxmath}
\usepackage{subfigure}
\usepackage{tikz}
\theoremstyle{definition}
\newtheorem{theorem}{Theorem}
\newtheorem{corollary}{Corollary}
\newtheorem{lemma}{Lemma}
\newtheorem{example}{Example}
\newtheorem{remark}{Remark}
\newtheorem{definition}{Definition}

\usepackage{cite}

\def\ind{\mathbbm{1}}
\usepackage{breqn}

\usepackage[T1]{fontenc}
\title{Spatial Network Calculus and Performance Guarantees in Wireless Networks}

\author{
\IEEEauthorblockN{Ke Feng\IEEEauthorrefmark{1} and Fran\c cois Baccelli\IEEEauthorrefmark{1}\IEEEauthorrefmark{2}}

\IEEEauthorblockA{\IEEEauthorrefmark{1}INRIA-ENS, Paris, France,
\IEEEauthorrefmark{2}Telecom Paris, France
 }
}
\begin{document}
\maketitle

\begin{abstract}

This work develops a novel approach toward performance guarantees for all links 
in arbitrarily large wireless networks. It introduces a \textit{spatial network calculus}, consisting of spatial regulation properties
for stationary point processes
and the first steps of a calculus for this regulation,
which can be seen as an extension to space of the classical network calculus. Specifically, two classes of
regulations are defined: one includes ball regulation and shot-noise regulation,
which are shown to be equivalent and upper constraint interference;
the other one includes void regulation, which lower constraints the signal power. These regulations
are defined both in the strong and weak sense: the former requires the regulations
to hold everywhere in space, whereas the latter only requires the regulations to hold as observed by a jointly stationary point process. 

Using this approach, we derive performance guarantees in device-to-device, ad hoc, and cellular networks under proper regulations.  We give universal bounds on the
SINR for all links, which gives link service guarantees
based on information-theoretic achievability. They are combined with classical network calculus to
provide end-to-end latency guarantees for all packets in wireless queuing networks.
Such guarantees do not exist in networks that are not spatially regulated, e.g., Poisson networks.
\end{abstract}


\section{Introduction}
Network calculus \cite{cruz1991calculus, 
chang2000performance,le2001network} is a (the) key tool for establishing deterministic latency
guarantees in wireline computer networks subject to such uncontrolled events
as bulk arrival of information to be processed, conjunction
of multiple accesses to the network, processing speed slowdowns due to preemptive
systems-level tasks, etc. Such guarantees are fundamentally needed
in mission-critical and real-time applications
where strict real-time is required, e.g., the control
network of an airplane or a car, that of an airport, a manufacturing floor,
or a nuclear plant. This is obtained by 1) shaping the input arrival processes
and 2) providing guaranteed service curves. Every computer or communication link being
a physical system composed of electronic devices subject to quantum physics,
strict determinism is of course only true up to certain limitations which are
of inherent probabilistic nature.  For this and many other and more compelling reasons,
a recent trend in this domain has been stochastic network
calculus, which accepts some forms of controlled randomness in the input and service
curves and replaces the deterministic guarantees by, e.g., a strict control of  
the tail distribution of latencies.

It is fair to say that wireless networking lags behind in terms of
provable guarantees compared to what wireline networks have been offering
to industry and embedded systems for more than 30 years:
despite the strong claims from the 5G and 6G industries that ultra-reliable
and low-latency communications (URLLC) will support 100\% of users in intended use cases \cite{3gpp38824,cavalcanti2019},
provable performance guarantees for \textit{all links} in an arbitrarily large wireless
network remain largely unavailable to the best of our knowledge.  The most relevant results so far are the meta distributions \cite{haenggi16meta,Feng20separability,kountouris18QoS}, which allow link-level characterizations in large networks without providing guarantees. 
There are at least two intrinsic reasons for this. The first set of reasons
comes from physics and more precisely electromagnetism.
Wireline propagation takes place in a controlled (and even designed) 
medium, whereas open space wireless propagation is subject to uncontrolled multi-path reflections
and direct path obstructions, two phenomena that make wireless link
characteristics fluctuate in an unpredictable way, and which are
captured by the statistical concepts of fading and shadowing, respectively. 
The second set of reasons comes from the shared nature of the wireless medium.
In contrast to what happens on a cable, in the open space,
there is no way to isolate a given electromagnetic transmission from
other, even distant, concurrent transmissions using the same spectrum. The interference
that other sources create hence plays a key role here and is even well
acknowledged to be the key limitation in dense networks. Frequency reuse and carrier sensing are in a sense natural but only heuristic ways to cope with this question \cite{al-hourani2019-frequency,Nguyen07-80211}. 
The control of multiple access in this open space context is still only
very partially mastered. 
All these phenomena (propagation, fading, shadowing, multiple access, and interference) are key ingredients
in the definition of the Shannon rate, which decides the achievable performance for the communication service process
of wireless links. Thus one understands why deterministic
or strict stochastic guarantees are not available yet. 

So far, performance analysis of large wireless networks relies on a well-recognized framework based on stochastic geometry and particularly, Poisson point processes (PPPs) \cite{Baccelli1997,Andrews2011Tractable,haenggi2009stochastic,haenggi2016meta,baccelli:hal-02460214}. This framework often allows computable and closed-form results, whose analogs in queueing theory are M/M/1 queues or more generally, Jackson networks. Unfortunately, none of these frameworks offers deterministic guarantees. In comparison, the present work shifts the focus from providing new closed forms for a specific spatial model to regulation properties, which, when implemented, lead to the desired bounds.


\subsection{Spatial Network Calculus}
The present paper proposes a new approach to this class of questions 
through the construction of a \textit{spatial network calculus} which is based on  spatial regulations that
are algorithmically implementable and which
provides guarantees for all links and their concatenation in wireless networks
of arbitrarily large size. For spatially regulated point processes, a key example is hardcore point processes. Given spatial regulations on transmitter and receiver processes, this calculus provides computable lower bounds on the signal-to-interference-plus-noise ratio (SINR)
for all links in the network that are deterministic in the case without fading and stochastic in the case with fading. This in turn provides service curves (i.e., a lower bound on the Shannon rate) on all links that are deterministic in the absence of fading and stochastic with fading. The service curves are further combined with classical network calculus toward bounded or controlled latency in wireless queueing systems.


\subsection{Summary of Spatial Regulations}
Spatial regulations of stationary spatial point processes are defined in Sections II and III.
 Two classes of spatial regulations are defined. The first class is that meant to
control the clustering of wireless links. This class contains ball and shot-noise regulations,
which are shown to be equivalent. 
The second class contains void regulation and is meant to control link distance.

Regulations are defined both in the strong and weak sense. Let $\Phi$ and $\Psi$ be jointly stationary and ergodic point processes on $\mathbb{R}^2$ \cite{baccelli:hal-02460214},  denoting the transmitter and receiver point process, respectively.
A \textit{strong regulation} of $\Phi$ is required to hold
everywhere in space, whereas a  \textit{weak regulation} of $\Phi$ with respect to $\Psi$ is required to  hold only as seen at the atoms of $\Psi$. The notion of weak regulation formalizes the notion of ``the typical observer'' based on Palm calculus \cite{baccelli:hal-02460214} and is natural in the wireless
setting. This was not discussed in classical network calculus and is new to the
best of our knowledge.



\subsection{Performance Guarantees in Wireless Networks}
The interest of spatial network calculus is illustrated by a few key wireless network architectures in Sections IV-VI, both of
 the device-to-device (D2D) and the cellular type.  The bipolar D2D architecture is considered in Section \ref{sec: d2d}, where transmitters form a stationary point process $\Phi$, and each transmitter has a dedicated receiver at a fixed distance $\tau$. 
 We show that the weak spatial regulation of $\Phi$ (with respect to the receiver point process) leads to a lower bound on the  SINR for  all links in the network. Further, assuming each transceiver pair has a queue
where the input arrival rate is equal to a constant $\lambda$ (the same for all queues) and the service rate is
the Shannon rate of the wireless link, we show that the spatial regulation of $\Phi$ leads 
to a situation where all queues are stable for $\lambda$ small enough,
whereas, in the absence of regulation, for all positive $\lambda$, there is a positive proportion
of the queues which are unstable. This is completed by deterministic latency guarantees
for all queues if the input arrival processes are time-regulated and there is no fading,
and by stochastic guarantees in terms of bounds on the tail otherwise.

Another classical D2D setting is the ad hoc one, considered in Section \ref{sec: ad-hoc}, where there is a single point 
process of transceivers $\Phi$. Assuming each node (transceiver) can establish links with those nodes w.r.t. which the SINR is above some predefined threshold bidirectionally, this defines a so-called SINR graph \cite{dousse05connectivity,baccelli2010stochastic}.
We show that if $\Phi$ is strongly void regulated and ball regulated, then one can operate long-range multi-hop
communications between arbitrarily distant nodes.
In case there is no fading, one can even
maintain a {\em wireless backbone} ({defined below}) in the ad hoc network,
capable of transmitting packets on arbitrarily large distances, with bounds on the latency
through the queues in series that compose the backbone. Again, these bounds are deterministic
if the input traffic arrival process is time-regulated and if there is no fading.
They are of stochastic nature otherwise.

The cellular setting considered in Section \ref{sec: cellular} features two jointly stationary point processes, a base station point process $\Phi$ and a user point process
$\Psi$.
We introduce yet another type of spatial regulation called
cell-load regulation, which controls the clustering of $\Psi$ in the Voronoi cells of $\Phi$. 
If $\Phi$ is weakly ball regulated w.r.t. $\Psi$, and  the network is cell-load regulated, then 
performance guarantees similar to those available for the previous network architectures hold for all users and their queues.

\section{Strong Spatial Regulations}

\subsection{Definitions}
Let $\Phi$ be a stationary and ergodic point process on $\mathbb{R}^2$, defined on the probability space $(\Omega,\mathcal{A},\mathbb{P})$\footnote{Following the standard notation in probability theory, $\Omega$, $\mathcal{A},~\mathbb{P}$ denote the sample space, the $\sigma$-algebra on $\Omega$, and the probability measure on $(\Omega,\mathcal{A})$, respectively.}. Let $b(x,r)\subset\mathbb{R}^2$ denote the open ball of radius $r>0$ centered at $x\in\mathbb{R}^2$ and $B(x,r)$ denote its closure. For a Borel set $L\subset\mathbb{R}^2$, let $\Phi(L)\in\mathbb{N}\triangleq\{0,1,2,...\}$ denote the number of points of $\Phi$ residing in $L$. Let $N(r)\triangleq \Phi(b(o,r))$ to simplify the notation. Throughout this work, we say that an event holds $\mathrm{a.s.}$ (which stands for almost surely)  to indicate that the probability of this event is 1.  



\begin{definition}[Strong $(\sigma,\rho,\nu)$-ball regulation]
\label{def: ball-reg}
A stationary point process $\Phi$ is strongly $(\sigma,\rho,\nu)$-ball regulated if  for all $r\geq0$,
\begin{equation}
    N(r)\leq \sigma +\rho r+\nu r^2,\quad \mathbb{P}\mhyphen \mathrm{a.s.}\label{eq: def-ball-reg}
\end{equation}
where $\sigma,\rho,\nu$ are constants and $\sigma,\nu\geq0$.
\end{definition}
Alternatively, we can write $\mathbb{P}(N(r)\leq \sigma +\rho r+\nu r^2)=1.$ By stationarity, for all $y\in\mathbb{R}^2$, $r>0$, $\Phi(b(y,r))$
is equally distributed as $\Phi(b(o,r))$ and is hence a.s. upper bounded by $\sigma +\rho r+\nu r^2$. 
A much stronger result is proved later in Lemma \ref{l2}, Section \ref{subsec: strong-weak}, which states that $\Phi$ 
is strongly $(\sigma,\rho,\nu)$-ball regulated if and only if $ \Pr(\cap_{y\in \mathbb{R}^2,r\geq0}
\Phi(b(y,r))\leq \sigma +\rho r+\nu r^2)=1$, i.e., the bound holds true for all locations for all $r$ simultaneously a.s.
 
This definition can generalize to ball regulation for stationary point processes on $\mathbb{R}^d$, 
provided that the polynomial of degree two is replaced by a polynomial of degree $d$.  
A special case is ball regulation for stationary point processes on $\mathbb{R}^+$. Let $\Phi$ be a stationary and ergodic point process
on $\mathbb{R}^+$. Define $(\sigma,\rho)$-ball regulation for $\Phi$ that for all $r>0$, $\Phi([0,r))\leq \sigma +\rho r,\mathbb{P}\mhyphen \mathrm{a.s.}$
By stationarity, for all $0\leq s\leq t$, $\Phi([s,t))\leq \sigma +\rho (t-s),~\mathbb{P}\mhyphen \mathrm{a.s.}$
On $\mathbb{R}^+$, ball regulation hence  boils down to the $(\sigma,\rho)$-regulation for packet arrivals in classical network calculus,
except that here the model is stochastic.

We say that a stationary point process $\Phi$ is strongly ball regulated if there exist some finite constants 
$\sigma,\rho,\nu$ such that $\Phi$ is strongly $(\sigma,\rho,\nu)$-ball regulated.




\begin{definition}[Strong $(\sigma,\rho,\nu)$-shot-noise regulation]
\label{def: sn-reg}
A stationary point process $\Phi$ is strongly $(\sigma,\rho,\nu)$-shot-noise regulated if, for all non-negative,
bounded, and non-increasing functions $\ell\colon \mathbb{R}^{+}\to\mathbb{R}^{+}$, and for all $R>0$,
\begin{equation}  
\label{eq: def-sn-reg}
\sum_{x\in \Phi\cap b(o,R)} \ell(\|x\|) \leq \sigma \ell(0)+ \rho \int_{0}^{R}{\ell(r)}\dd r+2 \nu \int_{0}^R r\ell(r)\dd r,~\mathbb{P}\mhyphen \mathrm{a.s.},\end{equation}
where $\sigma,\nu$ are positive constants.
\end{definition}
By stationarity, (\ref{eq: def-sn-reg}) is equivalent to saying that for all $ y\in\mathbb{R}^2$,
$\sum_{x\in \Phi\cap b(y,R)} \ell(\|x-y\|)\leq \sigma \ell(0)+ \rho \int_{0}^{R}{\ell(r)}\dd r+2 \nu \int_{0}^R r\ell(r)\dd r, \mathbb{P}\mhyphen \mathrm{a.s.}$    
Definition \ref{def: sn-reg} is a special case of Definition \ref{def: ball-reg}, as the latter requires 
(\ref{eq: def-sn-reg}) to hold for $\ell(\cdot)\equiv 1$ only.
\begin{remark}
Eq (\ref{eq: def-sn-reg}) is equivalent to
    \begin{equation}
\sum_{x\in \Phi}\ell(\|x\|)\leq A_{\ell}, \quad  \mathbb{P}\mhyphen \mathrm{a.s.},
\label{eq: sn-reg}
\end{equation}
where \begin{equation}
A_{\ell}\triangleq\sigma \ell(0)+ \rho \int_{0}^{\infty}{\ell(r)}\dd r+2 \nu \int_{0}^\infty r\ell(r)\dd r.
\end{equation}
$A_\ell$ is finite if and only if $\int_0^\infty r\ell\dd r<\infty.$
It is easy to see that (\ref{eq: def-sn-reg}) implies (\ref{eq: sn-reg}) by letting $R\to\infty$.
For all $R>0$, by letting $\ell(r) = \ind(r\leq R)$ (the indicator function that is equal to 1 if the argument 
is less than $R$ and 0 otherwise), one retrieves (\ref{eq: def-sn-reg}). The LHS of (\ref{eq: sn-reg})
is the shot-noise\footnote{The expectation of the shot-noise is known by Campbell's theorem
\cite{daley2008introduction}.} generated by $\Phi$ and $\ell$
at the origin, hence the name used in this definition. 
\end{remark}


The physical meaning of this regulation is related to the total received power and interference in wireless networks. In wireless networks, the path loss function naturally satisfy the conditions on $\ell$. 

\begin{theorem}[Equivalence]
A  stationary point process $\Phi$ is strongly
$(\sigma,\rho,\nu)$-shot-noise regulated if and only if it is strongly $(\sigma,\rho,\nu)$-ball regulated.
\label{thm: ball-sn-eq}
\end{theorem}
\begin{proof}

Firstly, the necessary part can be proved by letting $\ell(\cdot)\equiv {1}$. Now we prove the sufficient part. For $R>0$, let $\Delta \triangleq{R}/{n}$ and $r_k \triangleq k\Delta $, $n\in\mathbb{N},~k={0,...,n}$. Define the monotonic sequences $(B_0,B_1,...)$ and $(l_0,l_1,...)$, where $B_k \triangleq b(o,r_k),~l_k\triangleq \ell(r_k),~k=0,...,n$. The shot-noise generated by points in $b(o,R)$ at location $o$ is
\begin{align}
&\nonumber\sum_{x\in\Phi\cap b(o,R)} \ell(\|x\|)\\
              &\pleq{a} l_0 \Phi(B_1)+\sum_{k=1}^{n-1}l_k\big(\Phi(B_{k+1})-\Phi(B_k)\big)\nonumber \\
              &\peq{b} 
              \sum_{k=1}^{n} (l_{k-1}-l_{k})\Phi(B_{k}) + l_{n}\Phi(B_{n})\nonumber\\
              &\pleq{c} \sum_{k=1}^{n} (l_{k-1}-l_{k})(\sigma+\rho r_k+\nu r_k^2)+l_{n}(\sigma+\rho R+\nu R^2),
\quad\mathbb{P}\mhyphen \mathrm{a.s.}\nonumber        \end{align}
    Step (a) follows from the monotonicity of $\ell$. Step (b) follows from summation by parts. Step (c) follows from the fact that $\Phi$ is  strongly $(\sigma,\rho,\nu)$-ball regulated. The summation $\sum_{k=1}^{n} (l_{k-1}-l_{k})(\sigma+\rho r_k+\nu r_k^2)$ converges to the Riemman-Stieltjes integral $\int_{0}^{R}\sigma+\rho r+\nu r^2 \dd \ell$ as $n\to\infty$. The integral exists because $\ell$ is monotone and bounded, and thus has countable discontinuities. Now, using summation by parts, we retrieve (\ref{eq: def-sn-reg}).
\end{proof}
\begin{remark}
Theorem 1 shows that the shot-noise of strongly ball regulated point processes is upper bounded with computable bounds.
This theorem is applied in deriving deterministic bounds on the SINR in the case without fading for all links
in wireless networks (see, e.g., Section IV).
\end{remark}

The equivalence between strong ball regulation and strong shot-noise regulation is due to the universality of the function $\ell$ in Definition \ref{def: sn-reg}. For a fixed $\ell$, the conditions under which $\sum_{x\in\Phi}\ell(\|x\|)\leq A_\ell$ is weaker, as it does not imply that the inequality holds for other functions. However, this universality is important and has implications beyond the equivalence theorem as shown by what follows. Let $\{h_x\}_{x\in\Phi}$ denote the marks associated with $\Phi$. The shot-noise, $\sum_{x\in\Phi}h_x\ell(\|x\|)$, is now also subject to the randomness of $\{h_x\}_{x\in\Phi}$. Let $\mathcal{L}_{X}(s)\triangleq \Ex [\exp(-s X)]$ denote the Laplace transform of random variable $X$.  
\begin{definition}[Exponential moment condition\footnote{This is also known as the Cram\'er condition \cite{Cramr1994SurUN}.}]
\label{def: cramer}
 A non-negative random variable $X$ has exponential moments if  $\exists s^*>0$ such that its moment-generating function, $\mathcal{L}_{X}(-s)=\Ex [\exp(s X)]$, is finite for $s\in[0,s^*)$. 
\end{definition}

\begin{theorem}[Bounded conditional Laplace transform]
\label{thm: sn-laplace} 
Let $\{h_x\}$ be i.i.d. non-negative marks with finite mean and exponential moments.
Let ${\Tilde\ell}(r)= \log\mathcal{L}_h(-s\ell(r))$.
If $\Phi$ is strongly $(\sigma,\rho,\nu)$-ball regulated, then $\exists s^{*}>0$ such that for $s\in[0,s^{*})$,
\begin{equation}
 \Ex \left[\exp\left(s\sum_{x\in\Phi}h_x\ell(\|x\|)\right)~\bigg|~ \Phi\right]\leq \exp\left(A_{\Tilde\ell}\right),\quad  \mathbb{P}\mhyphen \mathrm{a.s.}
\end{equation}

\end{theorem}

\begin{proof}
\begin{align}
   \Ex  \left[\exp\left(\sum_{x\in\Phi}sh_x\ell(\|x\|)\right)~\Bigg|~ \Phi\right]\nonumber
   &\peq{a}\prod_{x\in\Phi} \mathcal{L}_h(-s\ell(\|x\|))\nonumber\\
&=\exp\left(\sum_{x\in\Phi}\log \mathcal{L}_h\left(-s\ell(\|x\|)\right)\right).\nonumber
\end{align}
Step (a) follows from the i.i.d. assumption of the marks. Then it follows from the equivalence
between $(\sigma,\rho,\nu)$-ball and shot-noise regulation,
if we can show that $\tilde\ell$ satisfies the path loss conditions in Definition \ref{def: sn-reg}.
Now, for $s\geq0$, $\tilde\ell$ is non-negative since $\Ex_h [\exp{(sh\ell)}]\geq 1$. Secondly,
$\exists s^*>0$ such that $\tilde\ell$ is bounded for $0\leq s\leq s^*$ due to the boundedness
of $\ell$ and the exponential moment condition on $h$. It is easy to see that $\tilde\ell(r)$ is non-increasing.
\end{proof}
\begin{remark}
Theorem 2 shows that the conditional Laplace transform of the shot-noise of an i.i.d. marked and
strongly ball regulated point process given the point process is upper bounded.
This result is applied in deriving lower bounds on reliability for all links where fading are i.i.d. marks with exponential moments. See Section \ref{sec: d2d}.
\end{remark}
\begin{remark}
	$A_{\tilde\ell}$ is finite if $\int_0^\infty r\ell(r)\dd r<\infty$ since $\log \mathcal{L}_h(-sh\ell(r))\sim s\ell(r) \Ex h,~r\to\infty$.
\end{remark}
\begin{definition}[Strong $\tau$-void regulation]
A stationary point process $\Phi$ is strongly $\tau$-void regulated if
$\Phi(B(o,\tau))\geq 1,\ \mathbb{P}\mhyphen \mathrm{a.s.}$, 
where $\tau$ is a positive constant.\end{definition}

By stationarity again, for all $y\in\mathbb{R}^2$, $\Phi(B(y,\tau))\geq 1$, $\mathbb{P}\mhyphen \mathrm{a.s.}$
Further, it is shown in Section \ref{subsec: strong-weak} that, if $\Phi$ is strongly $\tau$-void regulated,
then $\mathbb{R}^2$ can be a.s.  covered by $\bigcup_{x\in \Phi}B(x,\tau)$, the union of all closed 
disks centered at points in $\Phi$ with radius $\tau$.
The definition also generalizes to stationary point processes in $\mathbb{R}^d$.

The physical meaning of this regulation is related to the notion of coverage, e.g., in cellular or satellite networks.
Let $\Phi$ model the locations of transmitters in a wireless network. If $\Phi$ is strongly $\tau$-void regulated,
then for an arbitrary location in the network, there must be at least one transmitter within radius $\tau$,
which provides a lower bounded signal power in the absence of fading.

\subsection{Properties}
Let $\Phi$ be a stationary point process and $\Phi_i,~i=1,...,n$ be a family of stationary point processes on $\mathbb{R}^2$, defined on the probability space $(\Omega,\mathcal{A},\mathbb{P})$. The following properties hold:
\begin{itemize}
\item {(Superposition)}
    If $\Phi_i$ is strongly $(\sigma_i,\rho_i,\nu_i)$-ball regulated, $ i=1,...,n$, then the superposition  $\bigcup_{i=1}^{n}\{\Phi_i\}$ is strongly $(\sum \sigma_i, \sum \rho_i,\sum \nu_i)$-ball regulated. If at least one of $\Phi_i$ is not strongly ball regulated, then the superposition is not strongly ball regulated. If $\Phi_i$ is strongly $\tau_i$-void regulated, $i=1,...,n$, then the superposition  $\bigcup_{i=1}^{n}\{\Phi_i\}$ is strongly  $\min\{\tau_i\}$-void regulated.
     
\item {(Thinning)}
If $\Phi$ is strongly ball regulated, then it remains so after an arbitrary stationarity-preserving thinning. If $\Phi$ is not strongly void regulated, then it remains so after an arbitrary thinning. 

\item {(Displacement)} Let $\Tilde\Phi \triangleq \{x\in\Phi\colon x+V_x\}$ be a displaced point process of $\Phi$, where ${\{V_x\}}_{x\in\Phi}$ is a set of random vectors. If the random vectors $V_x$ are bounded and stationarity-preserving. Then $\Tilde\Phi$ is strongly ball regulated if and only if $\Phi$ is strongly ball regulated, and  $\Tilde\Phi$ is strongly void regulated if and only if $\Phi$ is strongly void regulated.
\end{itemize}

 The superposition property follows from the fact that the intersection of a countable collection of almost sure events is also almost sure. The thinning property is easy to see. For the displacement property, let $D$ denote an upper bound on the displacement vectors. Let $\Phi$ be strongly $(\sigma,\rho,\nu)$-ball regulated. $\Tilde N(r)\leq N(r+D)\leq \sigma+\rho(r+D)+\nu(r+D)^2, \mathbb{P}\mhyphen \mathrm{a.s.}$, which shows that $\Tilde\Phi$ is ball regulated by some $(\Tilde \sigma, \Tilde \rho, \nu)$. Conversely, $ N(r)\leq \Tilde N(r+D),~\mathbb{P}\mhyphen a.s$. Combining both inequalities proves the necessary and the sufficient condition.   Let $\Phi$ be $\tau$-void regulated, i.e., $\Phi(B(o,\tau))\geq1$. Then $\Tilde \Phi(B(o,\tau+D))\geq1$, and so $\Tilde \Phi$ is $\tau+D$-void regulated. The converse is true by a similar argument. Note that these bounds are the worst-case bounds, which may be improved when given additional assumptions about the point processes. 

These properties are useful in deriving bounds in wireless network models. For example, the superposition property gives SINR bounds in heterogeneous  cellular networks, where each layer of the network is strongly regulated by some $(\sigma,\rho,\nu)$. The displacement property relates to mobility. It implies that a strongly regulated network remains regulated under finite steps of bounded motion, potentially with a different bound.

The following results show local regulation of ball regulated point processes and the unboundedness of shot-noise in the absence of ball regulation.

\begin{lemma}[Local boundedness]A stationary point process $\Phi$ is strongly ball regulated if and only if there exist constants $H>0$ and $K<\infty$ such that $\Phi(b(o,H))\leq K$, $\mathbb{P}\mhyphen \mathrm{a.s.}$\label{lemma: ball-reg-local}
\end{lemma}
\begin{proof}
For the necessary condition, we know that there exists a $(\sigma,\rho,\nu)$ triple such that $\Phi$ is $(\sigma,\rho,\nu)$-ball regulated. Choose any $H>0$ and let $K=\sigma+\rho H+\mu H^2$. By Definition 1, $\Phi(b(o,H))\leq K,~\mathbb{P}$-a.s. For the sufficient condition, let $\phi$ denote a centered deterministic triangular lattice of intensity $2/(3\sqrt3 H^2)$. For all $r\geq0$, $\bigcup_{y\in\phi\cap b(o, r)} b(y,H)$ is a finite cover of $b(o,r)$. We have
$\Phi(b(o,r)) \leq \Phi\left(\bigcup_{y\in\phi\cap b(o, r)} b(y,H)\right)
    \leq \sum_{y\in\phi\cap b(o, r)} \Phi(b(y,H)) \leq K \phi(b(o,r)),~  \mathbb{P}\mhyphen \mathrm{a.s.}\nonumber$
 The last step follows from the facts that $\Pr(\Phi(b(y,H)) \leq K) =1$, $\forall y\in\mathbb{R}^2$, by stationarity, and that the intersection of a finite collection of events is a.s. true if each of the events is a.s. true. It is easy to see that for all $r>0$, $\phi(b(o,r))$ is upper bounded by some deterministic polynomial of $r$ of degree 2. This proves the existence of  a $(\sigma,\rho,\nu)$ triple.
\end{proof}
This lemma leads to a constructive way to implement ball regulation, which is through local repulsion. For instance, a hardcore point process with hardcore distance $H>0$ is one where no two points coexist in any disk of radius $2H$. Then it follows immediately from Lemma \ref{lemma: ball-reg-local} that hardcore point processes are ball regulated with $K=1$. Specially, any stationary lattice is ball regulated. These examples are gathered in the next subsection.

 \begin{lemma}
  Let $\Phi$ be a stationary point process. Let $\ell\colon\mathbb{R}^+\to\mathbb{R}^+$ be non-negative, non-increasing, and with non-empty essential support. If $\Phi$ is not strongly ball regulated, then $\sum_{x\in\Phi}\ell(\|x\|)$ cannot be bounded from above a.s.
\end{lemma}
\begin{proof}
Choose $\epsilon>0$ such that $R(\epsilon) \triangleq \sup\{r\colon\ell(r)\geq\epsilon\}>0$.  Since $\Phi$ is not strongly ball regulated, by Lemma \ref{lemma: ball-reg-local}, $\Phi(b(o,R(\epsilon)))$ cannot be bounded from above by a constant $\mathrm{a.s.}$ Hence $\sum_{x\in\Phi}\ell(\|x\|)\geq \sum_{x\in\Phi\cap b(o,R(\epsilon))}\ell(\|x\|)\geq \epsilon \Phi(b(o,R(\epsilon)))$  cannot be bounded from above a.s.
\end{proof}

 This lemma shows that ball regulation is a necessary and sufficient condition for the a.s. boundedness of shot-noise. Consequently, the conditional Laplace transform of an i.i.d. marked point process cannot be a.s. bounded in the absence of ball regulation.
 

\subsection{Examples}
\begin{table}[t]
\begin{center}
\caption{Strong regulation properties of some common point processes. }
\begin{tabular}{|c|  c c|} 
 \hline
\bfseries Stationary PPs & \bfseries ball regulation & \bfseries void regulation  \\ [0.5ex] 
 \hline
 \hline
 PPP & \text{\sffamily x} &  \text{\sffamily x} \\ 
 \hline
 Neyman-Scott cluster process & \text{\sffamily x} & \text{\sffamily x}  \\ 
 \hline
 Lattices & \checkmark & \checkmark  \\
 \hline
 Perturbed lattices & \checkmark & \checkmark \\
 \hline
 Mat\'ern hardcore PP & \checkmark & \text{\sffamily x}  \\
 \hline
 PPP+Lattices & \text{\sffamily x} &\checkmark \\ 
 \hline
\end{tabular}
\label{table:1}
\end{center}
\end{table}

This subsection gathers a few examples illustrating the definitions so far.




\begin{corollary}
 The Poisson point process is neither strongly void nor strongly ball regulated.
\end{corollary}
\begin{proof}
For any $r>0$, the number of points in $b(o,r)$ or $B(o,r) $ is Poisson distributed, which cannot be a.s. bounded from above or bounded away from 0. Then it follows from Lemma \ref{lemma: ball-reg-local}.
\end{proof}
A Neyman-Scott cluster process is a cluster process where the parent point process is a PPP, and each cluster is i.i.d. and random in number. It is easy to show that this process is neither ball or void regulated. 

\begin{lemma}
Any hardcore point process on $\mathbb{R}^2$ with hardcore distance $H$ is strongly
$(1,2\pi/(\sqrt{12}H),\pi/(\sqrt{12}H^2))$-ball regulated.
\label{lemma: hc}
\end{lemma}
\begin{proof}
    See Appendix \ref{appendix: hardcore}.
\end{proof}

It suffices to remove a small fraction of points from a stationary PPP  to make it  hardcore, thus strongly ball regulated. The resulting point processes are known as Mat\'ern hardcore processes. A Mat\'ern hardcore process of type I with hardcore distance $H$ is obtained from a PPP by retaining a point if there is no other point within distance $2H$ from it and not retaining it otherwise. A Mat\'ern hardcore processes of type II with hardcore distance $H$ is obtained from a PPP where an independent random mark is associated with each point of the  PPP, and a point is retained if there is no point within distance $2H$ from it with a bigger mark.  Hardcore point processes with higher densities can be realized through sequential inhibition processes.
 




\begin{example}[Mat\'ern hardcore processes]

Mat\'ern hardcore processes are strongly ball regulated for being hardcore. They are not strongly void regulated due to the thinning property and the fact the PPP is not strongly void regulated.
\end{example}

 By the superposition property, the superposition of a finite collection of Mat\'ern hardcore point processes is still ball regulated (but not necessarily hardcore).   
 
\begin{example}[Lattices and perturbed lattices]
Lattices are both strongly ball and strongly void regulated.  By the displacement property, independently disturbed lattices with bounded displacement are also strongly void regulated and ball regulated.
 \end{example}
\begin{example}[Superposed PPP and lattice]
The superposition of a PPP and a lattice is strongly void regulated but not strongly ball regulated. This follows from the superposition property.
\end{example}
\begin{example}[Cluster processes]Let $\Phi_{\mathrm{p}}$ be a stationary point process on $\mathbb{R}^2$ and $\{\Phi_{x}\}_{x\in\Phi_{\mathrm{p}}}$ be a family of a.s. bounded points sets distributed in bounded balls. Let the cluster point process be $\Phi=\bigcup_{x\in\Phi_{\mathrm{p}}}\{x+\Phi_{x}\}$. If $\Phi_{\mathrm{p}}$ is strongly ball regulated,  then $\Phi$ is strongly ball regulated.
\end{example}

Table I gives a summary of the strong regulation properties discussed in the examples, where ``+'' indicates superposition. 
  \subsection{Extremal Regulation Parameters}
Let $\mathcal{D}_{}\subset\mathbb{R}^3$ be defined as $
    \mathcal{D}_{}\triangleq\{z\in\mathbb{R}^3\colon \Phi~{\rm{is}}~z{\rm{-ball~regulated}}\}.$
$\mathcal{D}$ is the set of regulation parameters associated with the ball regulation property of $\Phi$. $\mathcal{D}$ is convex and  unbounded.  To see this, let $\Phi$ be strongly $(\sigma_1,\rho_1,\nu_1)$-ball regulated and $(\sigma_2,\rho_2,\nu_2)$-ball regulated. We can easily show that for all $p\in[0,1]$,
$p(\sigma_1, \rho_1,\nu_1)+(1-p)(\sigma_2,\rho_2,\nu_2)\in\mathcal{D}$, $\forall p\in[0,1]$, which shows the convexity. The unboundedness is trivial as each direction can grow arbitrarily large. A question that naturally arises is the extremal parameters of $\mathcal{D}$. Here we look at two separate aspects: one defines the extremal parameters which achieve the infimum in each infinite direction of $\mathcal{D}$; the other is to define the extremal triple jointly.

Let $\Phi$ be a strongly ball regulated process and $\mathcal{D}$ be defined as above.  Let
$\sigma_c \triangleq \inf \{\sigma\colon (\sigma,\rho,\nu)\in\mathcal{D}, {\rho, \nu<\infty}\};~
\rho_c \triangleq \inf \{\rho\colon (\sigma,\rho,\nu)\in\mathcal{D}, {\sigma, \nu<\infty}\};$
$\nu_c \triangleq \inf \{\nu\colon (\sigma,\rho,\nu)\in\mathcal{D}, {\sigma, \rho<\infty}\}.
$
For non-degenerate point processes, i.e., $\Phi\neq\emptyset, \mathbb{P}\mhyphen \mathrm{a.s.}$, $\sigma_c \geq 1$, {which accounts for the fact that the contact distance} $\min\{x\in\Phi\ \colon\|x\|\}$ can be arbitrarily small. $\rho_c = 0$ since one can always find $\sigma,\nu$ large enough such that $\rho R$ is upper bounded by $\sigma+\nu R^2$, for all $R>0.$ $\nu_c$ is non-trivial and depends on the distribution of $\Phi$. Later we will show that, in general, we need to choose $\rho>0$ in order to achieve $\nu_c$. 
For a fixed $\ell$, the extremal shot-noise regulation  is $\inf_{(\sigma,\rho,\nu)\in\mathcal{D}} {A_\ell}$.

For a strongly void regulated process, we define its extremal void regulation parameter to be
$\tau_c \triangleq \inf \{\tau\colon {\Phi(B(o,\tau))\geq 1,~\mathbb{P}\mhyphen \mathrm{a.s.}}\}.$ Here
$\tau_c$ and $\nu_c$ give the following lower and upper bounds on the density of a strongly ball and void-regulated point process, i.e.,
\begin{equation}
    1/\tau_c^2 \leq \Ex \Phi(B(o,1)) \leq \nu_c. \label{eq: joint-regulation}
\end{equation}
The right inequality is easy to see. To see the left inequality, $\Ex \Phi(B(o,1)) =  \Ex \Phi(B(o,\tau_c))/\tau_c^2 \geq \Phi(B(o,\tau_c))/\tau_c^2\geq1/\tau_c^2$. Eq (\ref{eq: joint-regulation}) shows that strong $(\sigma,\rho,\nu)$-ball regulation and $\tau$-void regulation are not always jointly possible. For $\nu<\tau$, the joint regulation is not possible.
When $\nu_c = \Ex \Phi(B(o,1))$, $\sigma+\rho r$ upper constrains the deviation of the number of points around its mean for all $r>0$. 
It is easy to see that for any stationary lattice with intensity $\lambda$, $\nu_c= \lambda\pi.$

Extremal regulation parameters for general ball regulated point processes lead to tighter upper bounds for the shot-noise, which relates to interference power in wireless networks. We will show, by the generalization that follows, that the bounds may be further improved.


  \subsection{Strong $g$-Ball Regulation}
 The strong $(\sigma,\rho,\nu)$-ball regulation is based on an upper bound on the counting measure of a stationary point process in terms of a polynomial of degree 2. As stated in Lemma \ref{lemma: ball-reg-local}, all locally bounded point processes are ball regulated. For some point processes, it is known that non-polynomial form upper bounds exist, which are $O(r^2)$  but asymptotically tighter. For example, the upper bounds for the number of lattice points within a ball was first studied by Gauss, known as the Gauss circle problem. For a centered square lattice with unit density, Gauss proved that $N(r)-\pi r^2\leq 2\sqrt{2}r$. For general lattices  with shifted centers, it is shown that $N(r)-\pi r^2= O(r^{2/3})$ as $r\to\infty$ and conjectured to be  $O(r^{1/2+\epsilon})$ for all positive $\epsilon$ \cite{Kendall1948,LAX1982280}. This prompts the following generalization of Definition  \ref{def: ball-reg}. 
Let $g:\mathbb{R}^+\to\mathbb{R}^+$ be a non-negative and wide-sense increasing function with countable discontinuities, and such that $g(r)=O(r^2)$ as $r\to\infty$. 
  \begin{definition}[Strong $g$-ball regulation]
  \label{def: G-ball-reg}
A stationary point process $\Phi $ is strongly $g$-ball regulated if for all $r\geq0$, $
    N(r)\leq g(r),~\mathbb{P}\mhyphen \mathrm{a.s.}\label{eq: G-ball-reg}$  \end{definition}
By stationarity, for all $y\in\mathbb{R}^2$, $\Phi(b(y,r))$ is equally distributed as $\Phi(b(o,r))$. So the same a.s. upper bound holds on $\Phi(b(y,r))$.  A stronger result states that $ \Pr(\cap_{y\in \mathbb{R}^2,r\geq0}
 \Phi(b(y,r))\leq g(r))=1$. The formal proof are given in Section  \ref{subsec: strong-weak}. 

\begin{theorem}[Equivalence]
\label{thm: g-equivalence}
A stationary point process $\Phi$ is
strongly $g$-ball regulated if and only if for all non-negative and non-increasing 
 continuous functions $\ell\colon\mathbb{R}^+\to\mathbb{R}^+$ and for all $R>0$,
\begin{equation}
\sum_{x\in \Phi\cap b(o,R)} \ell(\|x\|) \leq  \int_{0}^{R}g(r)\dd \ell+\ell(R)g(R),~  \mathbb{P}\mhyphen \mathrm{a.s.}\label{eq: G-sn-reg}
\end{equation}
\end{theorem}
\begin{proof}
The proof follows from replacing $\sigma+\rho r+\nu r^2$ by $g(r)$ in the proof of  Theorem \ref{thm: ball-sn-eq}. $\int g \dd \ell$ is the Riemann-Stieltjes integral, which exists as $g$ has countable discontinuities.
\end{proof}

\begin{remark}
Taking the limit $R\to\infty$ of both sides of (\ref{eq: G-sn-reg}), provided that the RHS is finite, we obtain the $g$-shot-noise regulation.
\end{remark}
\begin{remark}
\label{remark: ell_condition}
     $\ell$ does not need to be bounded for generalized ball regulation. For example, the RHS of (\ref{eq: G-sn-reg}) exists  for $\ell(r) = r^{-\alpha}, \alpha>2$ provided that $g(r)$ is zero around the origin.
Here we require $\ell$ to be continuous only to make sure that $\int g\dd \ell$ exists. The condition can be relaxed by requiring that $\ell$ and $g$ share no common discontinuities. The properties of $(\sigma,\rho,\nu)$-ball regulation hold for $g$-ball regulation and are omitted here.
\end{remark}

\section{Weak Spatial Regulations}
In many applications, the perspective of interest is only that from a specific set of observers
rather than from the entire space. For example, in cognitive radio networks, there are a set of
primary transceivers and a set of secondary transceivers. The mandated regulation (or protection)
for the former can be more strict than that seen from an arbitrary location or the secondary transceivers.
This motivates what we propose to call \textit{weak spatial regulations}.
The ``strong-weak'' terminology will become clear by the end of this section. Weak spatial regulations build upon Palm 
calculus, which concern the perspective of the typical observer \cite{baccelli:hal-02460214}. 
  
   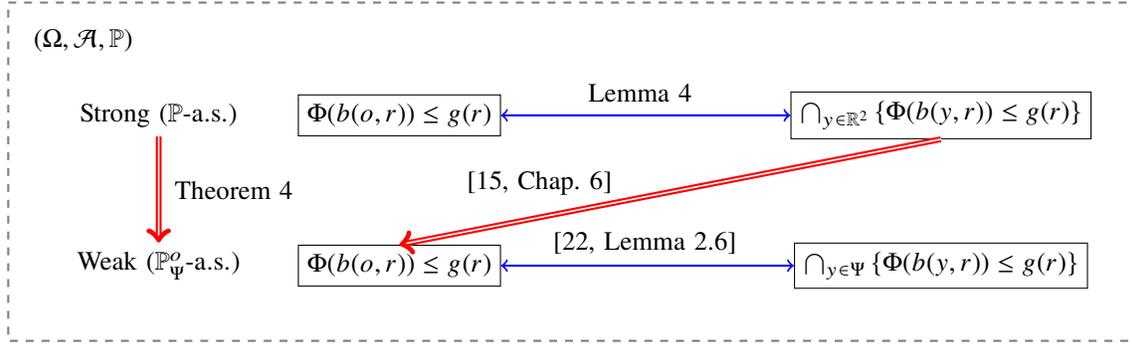
\begin{figure*}[t]
\centering
\normalsize
\begin{tikzpicture}
\node (space)   at (-4,3){$(\Omega,\mathcal{A},\mathbb{P})$};
\node (strong)   at (-3,2)   {Strong ($\mathbb{P}\mhyphen \mathrm{a.s.}$)};
\node (Implies)   at (-2,1) {Theorem \ref{thm: strong-weak}};
\node (eq2)   at (3.4,0.3) {\cite[Lemma 2.6]{khezeli:hal-03927500}};
\node (p-p0)   at (2,1.1) {\cite[Chap. 6]{baccelli:hal-02460214}};
  \node at (3.4,2.3)  {Lemma \ref{thm: measure-ball-reg}};
  \node (strong-all)   at (7.4,2)  [draw] {\texttt{$\bigcap_{y\in \mathbb{R}^2}\big\{
 \Phi(b(y,r))\le g(r)\big\}$}};
  \node (strong-o)    at (0.2,2) [draw] {\texttt{$ \Phi(b(o,r))\le g(r)$}};
  \node (weak)   at (-3,0) {Weak ($\mathbb{P}_\Psi^o\mhyphen \mathrm{a.s.}$)};
  \node (weak-all)  at (7.4,0)  [draw] {\texttt{$\bigcap_{y\in \Psi}\big\{
 \Phi(b(y,r))\le g(r)\big\}$}};
  \node (weak-o) at (0.2,0)  [draw] {$ \Phi(b(o,r))\le g(r)$};
  \draw [thick,blue,<->](node cs:name=weak-all,anchor=west)
      --(node cs:name=weak-o,anchor=east);
      \draw [thick,red,->,double](node cs:name=strong,anchor=south)
      --(node cs:name=weak,anchor=north);
             \draw [thick,red,->,double](node cs:name=strong-all,anchor=south)
      --(node cs:name=weak-o,anchor=north);
  \draw [thick,blue,<->](node cs:name=strong-all,anchor=west)
      --(node cs:name=strong-o,anchor=east);
       \draw [thick,gray,dashed] (-5,3.5) -- (10,3.5) -- (10,-1) --(-5,-1)--(-5,3.5);
\end{tikzpicture}
\vspace{.5em}
\caption{An illustration of the strong v.s. weak ball regulation relation. The diagram holds for all $r>0$. The blue line denotes equivalence between the events. The red double arrow  denotes ``implies''. Similar results hold for the strong v.s. weak void regulation and are omitted here.}
\label{fig: str-weak}
\vspace*{4pt}

\end{figure*}

    \subsection{Definitions}
    Let $\Phi$ and $\Psi$ be two point processes on $\mathbb{R}^2$ which are jointly stationary and ergodic, defined on the same probability space $(\Omega, \mathcal{A},\mathbb{P})$. Without loss of generality, let $\Psi$ be the set of observers, from which $\Phi$ are measured. Let $\mathbb{P}_{\Psi}^o$, $\mathbb{E}_{\Psi}^o$ denote the Palm probability and Palm expectation of $\Psi$, respectively.  Note that nothing forbids to take $\Psi=\Phi$.

Recall that $N(r)$ counts the number of points in $\Phi$ that reside in $b(o,r)$.     Let $g:\mathbb{R}^+\to\mathbb{R}^+$ be non-decreasing with countable discontinuities such that $g(r)=O(r^2)$ as $r\to\infty$. 
\begin{definition}
[Weak $g$-ball regulation]
A stationary point process $\Phi$ is (weakly) $g$-ball regulated with respect to a jointly stationary point process $\Psi$ if for all $r\geq0$,
$N(r)\leq g(r),~  \mathbb{P}_{\Psi}^{o}\mhyphen \mathrm{a.s.}$
\end{definition}
\begin{remark}
Regulation is in the weak sense when it is specified w.r.t. an observer point process. If $\Phi$ is independent of $\Psi$, then $\Pr_{\Psi}^{o} = \Pr$ \cite[Lemma 6.3.1]{baccelli:hal-02460214}, and we retrieve strong regulations as in Section II.  
\end{remark}

We say that $\Phi$ is \textit{$g$-shot-noise regulated with respect to $\Psi$} if for all $R>0$ and all non-negative and non-increasing functions $\ell\colon \mathbb{R}^{+}\to\mathbb{R}^{+}$, $\sum_{x\in\Phi\cap b(o,R)}\ell(\|x\|) \leq  \int_{0}^{R}g\dd \ell+\ell(R)g(R),~  \mathbb{P}_\Psi^o\mhyphen \mathrm{a.s.}$, provided that the Riemman-stietjes integral $\int g\dd \ell$ exists. Specially, $\Phi$ is $(\sigma,\rho,\nu)$-ball regulated with respect to $\Psi$ if for all non-negative, bounded, and non-increasing functions $\ell\colon \mathbb{R}^{+}\to\mathbb{R}^{+}$,
$\sum_{x\in\Phi} \ell(\|x\|)\leq\sigma \ell(0)+ \rho \int_{0}^{\infty}{\ell(r)}\dd r+2 \nu \int_{0}^\infty r\ell(r)\dd r,~\mathbb{P}_\Psi^o\mhyphen \mathrm{a.s.}$
\begin{corollary}[Equivalence]
A stationary point process $\Phi$ is 
$g$-shot-noise regulated with respect to a jointly stationary point process $\Psi$ if and only if it is $g$-ball regulated with respect to $\Psi$.
\label{thm: observer-G-ball-sn-eq}
\end{corollary}
\begin{proof}
    Apply the techniques in the proof of Theorem \ref{thm: g-equivalence}.
\end{proof}

\begin{definition}[Weak $\tau$-void regulation]
A stationary point process $\Phi$ is $\tau$-void regulated with respect to a jointly stationary point process $\Psi$ if
$\Phi(B(o,\tau))\geq 1,\ \mathbb{P}_{\Psi}^{o}\mhyphen \mathrm{a.s.}$,
where $\tau>0$ is a constant.
\label{def: ob-void-reg}
\end{definition}
\begin{remark}
    Any stationary point process is void regulated w.r.t. itself.
\end{remark}

The superposition property of strong spatial regulation holds for weak regulation provided
it is with respect to the same observer point process.
However, not all properties of strong regulations extend to weak regulations. 
For example, if there exist some constant $R,K>0$ such that $\Phi(b(o,R))\leq K$,
$\mathbb{P}_\Psi^o$-a.s. (referred to as local boundedness in Lemma \ref{lemma: ball-reg-local}),
$\Phi$ is not necessarily ball regulated with respect to $\Psi$.
Further, the absence of weak regulation does not imply that the shot-noise cannot be $\mathbb{P}_\Psi^o$-a.s.
upper bounded. One can construct such examples using fixed $\ell$ with bounded support.

\subsection{Relation to Strong Regulations}
This subsection justifies the strong-weak terminology used in the present paper.
\label{subsec: strong-weak}
Firstly,
we justify the claim that the strong regulation defined
above in terms of properties of the balls of center $o$
is in fact a property enforced everywhere in space. 
 \begin{lemma}
\label{l2}
For all fixed $l\in \mathbb{N}$ and $r>0$, 
the set $
\bigcap_{y\in\mathbb{R}^2} \{\Phi(b(y,r)\leq l\}$
is an event of $\mathcal A$. The set
$E_\Phi= \bigcap_{y\in \mathbb{R}^2,r\geq0}\big\{
 \Phi(b(y,r))\le g(r)\big\}$
is also an event of $\mathcal A$. Further, 
$\Pr(E_\Phi)=1$ if and only if for all $r\geq0$, $\Pr(\Phi(b(o,r))\le g(r))=1$.
\label{thm: measure-ball-reg}
\end{lemma}
\begin{proof}
See Appendix \ref{appendix: thm-measure-ball-reg}.
\end{proof}
\begin{lemma}
The set
$V_\Phi= \bigcap_{y\in \mathbb{R}^2}\big\{\Phi(B(y,\tau))\ge 1\big\}$
is an event of $\mathcal A$.
In addition, $\Pr(V_\Phi)=1$ if and only if $\Pr(\Phi(B(o,\tau))\ge 1)=1$.
\label{thm: measure-void-reg}
\end{lemma}
\begin{proof}
See Appendix \ref{appendix: thm-measure-void-reg}.
\end{proof}


Note that the events $\bigcap_{y\in\mathbb{R}^2} \{\Phi(b(y,r)\leq l\}$, $E_\Phi$, and $V_\Phi$ are $\{\theta_t\}$-invariant\cite[Chap. 6]{baccelli:hal-02460214}; such events hold $\mathbb{P}^o_\Psi\mhyphen \mathrm{a.s.}$ if they hold $\mathbb P\mhyphen \mathrm{a.s.}$

\begin{theorem}[Strong v.s. weak]
\label{thm: strong-weak}
Let $\Phi,~\Psi  $ be two jointly stationary point processes.
\begin{enumerate}
    \item  If $\Phi$ is strongly $g$-ball regulated, then all jointly stationary point processes $\Psi$, $\Phi$ is weakly $g$-ball regulated with respect to $\Psi$; if $\Phi$ is weakly $g$-ball regulated with respect to some $\Psi$, it is not necessarily strongly $g$-ball regulated.
\item If $\Phi$ is strongly $\tau$-void regulated, then for all jointly stationary point processes $\Psi$, $\Phi$ is weakly $\tau$-void regulated with respect to $\Psi$; if $\Phi$ is weakly $\tau$-void regulated with respect to some $\Psi$, it is not necessarily strongly $\tau$-void regulated. 
\end{enumerate}

\end{theorem}
\begin{proof}
For the first statement of 1), it follows from Lemma \ref{l2} that if $\Phi$ is strongly $g$-ball regulated, then $\Pr(E_\Phi)=1$.  Then it follows from the $\{\theta_t\}$-invariance of $E_\Phi$ that $\Pr_\Psi^o(E_\Phi)=1$, which implies the weak $g$-ball regulation. 
 For the second statement, it suffices to provide the following counter-example:
let $\Phi$ be a stationary square lattice with unit density. Let $\Psi = \Phi_{[1/2,0]}$ be  the observer point process, which is $\Phi$ shifted by the vector $(1/2,0)$. Consider the observer $(\sigma,\rho,\nu)$-ball regulation with respect to $\Psi$: $\sigma = 0$ belong to the set of regulation parameters of this observer point process. In contrast, the same parameters do not regulate $\Phi$.  

The first statement of 2) follows from  Lemma \ref{thm: measure-void-reg} and the $\{\theta_t\}$-invariance of $V_\Phi$, and the fact that $\Pr_\Psi^o(V_\Phi)=1$ implies the weak void regulation. For the second statement, let $\Phi$ be a stationary Poisson point process and $\Psi=\Phi_{[1/2,0]}$. Then $\Phi$ is 1/2-void regulated with respect to  $\Psi$ but not strongly void regulated.
\end{proof}

 An illustration of the key elements and their relations related to strong v.s. weak $g$-ball regulation is given in Fig. \ref{fig: str-weak}. It holds for all $r>0$. A similar relation holds for void regulation.

\section{Performance Guarantees in Device-to-device Bipolar Networks}
\label{sec: d2d}

\subsection{System Model}
Let the transmitter locations in a device-to-device (D2D) bipolar network \cite{baccelli2010stochastic} be modeled by a stationary and ergodic point process $\Phi$ on $\mathbb{R}^2$, defined on the probability space $(\Omega,\mathcal{A},\mathbb{P})$\footnote{Ergodicity, roughly speaking, means that one single realization of the point process contains the full statistical information. Stationarity and ergodicity allow one to simulate only one large spatial realization of the network.}. Each transmitter has a dedicated receiver at a fixed distance $\tau>0$ in a uniformly random direction.  The resulting receiver point process is denoted by $\Psi$, which is jointly stationary and ergodic with $\Phi$. 
 We focus on the typical D2D link (transmitter-receiver dipole), where the receiver is at the origin and its dedicated transmitter is denoted by $t\in\Phi$.
By the displacement property from Section II, $\Phi$ and $\Psi$ are either both strongly ball (void) regulated or neither ball (void) regulated. Further, by the construction of the model, $\Phi$ is weakly $\tau$-void regulated w.r.t. $\Psi$ (the reverse is also true).

\begin{figure}[t]
    \centering
    \centering
\includegraphics[width=.9\linewidth]{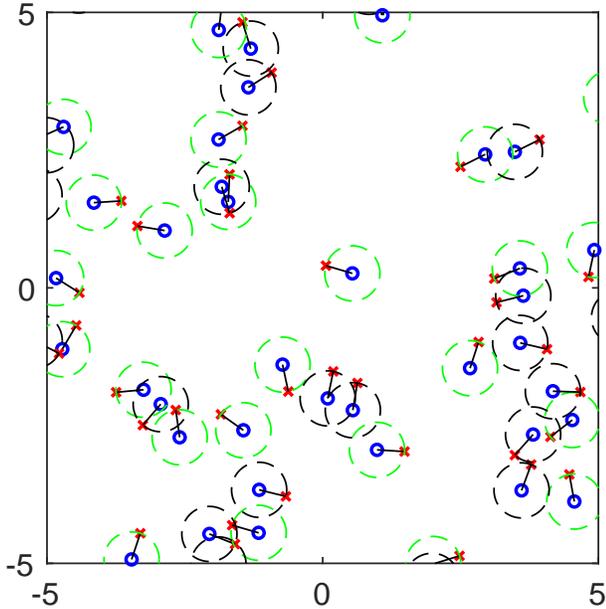}
    \caption{Illustration of a D2D bipolar network with $\tau=0.5$. Blue circles denote transmitter locations simulated by a PPP with intensity 0.5, and red x-marks denote receiver locations. Dashed circles represent guard zones centered around the transmitters with radius  0.5. Through thinning the PPP to obtain a Mat\'ern hardcore process of type II, only transmitters with green guard zones are retained. 
    }
    \label{fig: d2d}
\end{figure}

We first consider the case where all links are always active, i.e., full buffer at all transmitters. This assumption yields worst-case bounds. The queueing aspect is considered in the last subsection.  Without loss of generality, we assume that the bandwidth is normalized to 1.  We assume that the path loss is a function of distance denoted by $\ell: \mathbb{R}^+\to\mathbb{R}^+$. We assume $\ell$ is bounded, non-increasing, continuous, and integrable in $\mathbb{R}^2$, i.e., $\int_{0}^{\infty} \ell(r) r\dd r <\infty$. The boundness of $\ell$ is not strictly necessary for $g$-ball regulation as noted in Remark \ref{remark: ell_condition}. Without loss of generality, let $\ell(0)=1$. Let $W$ denote the variance of the additive white Gaussian noise. Denote by $h_x$ the power of the small-scale fading from $x\in\Phi$ to the origin. Let $f_h$, $F_h$, and $\mathcal{L}_{h}$ denote the PDF, CDF, and Laplace transform of a random variable $h$, respectively. Assume that the fading $\{h_x\}_{x\in\Phi}$ are i.i.d. with unit mean and satisfy the exponential moment condition (Definition \ref{def: cramer}).   Our framework generalizes easily to system models with shadowing, which will not be discussed here.


Fig. \ref{fig: d2d} illustrates a D2D bipolar network, where the original transmitters are distributed as a PPP, which is not ball regulated, and the retained transmitters are distributed as a Mat\'ern hardcore process of type II, which is ball regulated. The latter is realized through a thinning of the former.

\subsection{Performance Metrics}
The SINR at the origin is 
    \[
    \mathrm{SINR} =\frac{h_t\ell(\tau)}{I_t+W},\]
where $\|t\|=\tau$ and
$ I_t\triangleq\sum_{x\in \Phi\setminus\{t\}} h_x \ell(\|x\|)$ is the power of the total received  interference. For a target SINR threshold $\theta>0$, the link reliability is defined as
\[P_{\mathrm{s}}(\theta)\triangleq\mathbb{P}_{\Psi}^o(\mathrm{SINR}>\theta\mid \Phi).\]
The random variable $P_{\mathrm{s}}(\theta)\in[0,1]$ is the reliability of the typical receiver under fading conditioned on the locations of transmitters.
The distribution associated with $P_{\mathrm{s}}(\theta)$ is known as the \textit{meta distribution} \cite{haenggi16meta}, which gives the fraction of links in the network that achieve a link reliability above $x$ for a target SINR threshold $\theta$.  For some point processes, e.g., Poisson networks, the fraction is strictly less than 1, and its asymptotics is known \cite{haenggi16meta,Feng20separability}. The link  ergodic rate when treating interference as noise is
\[C\triangleq\Ex_{\Psi}^o[\log(1+{\rm{SINR}})\mid \Phi],\]
which is averaged over fading (in [nats/s/Hz]).

The regime of interest here is when $P_{\mathrm{s}}(\theta)$ and $C$ are a.s. lower bounded, i.e., all links achieve a link reliability and ergodic rate. Note that the former implies the latter.
\begin{definition}[Link performance bounds]
    For a target SINR threshold $\theta$, we say that $x$ is a {\em{reliability lower bound} for all links} if 
$\mathbb{P}_{\Psi}^{o} (P_{\mathrm{s}}(\theta)>x)=1.$
We say that $x$ is an {\em{ergodic rate lower bound for all links}} if $\Pr_{\Psi}^{o} \left(C\geq x\right) =1.
$
\end{definition}

\subsection{No Fading}
In the absence of fading, the link performance is determined by the network geometry. This is a degenerate case where the link reliability $P_{\rm{s}}(\theta)\in\{0,1\}$. The total interference is simply the Palm version of the shot-noise generated by $\Phi$ w.r.t. $\Psi$ excluding the shot from the dedicated transmitter $t$. Recall from Remark 1 that $A_{\ell}=\sigma \ell(0)+ \rho \int_{0}^{\infty}{\ell(r)}\dd r+2 \nu \int_{0}^\infty r\ell(r)\dd r$.
    \begin{corollary} \label{corr: D2D_wo_fading}
      If $\Phi$ is $(\sigma,\rho,\nu)$-ball regulated with respect to $\Psi$, in the absence of fading, \begin{equation}
        I_t\leq A_\ell -\ell(\tau), \quad  \mathbb{P}_{\Psi}^{o}\mhyphen \mathrm{a.s.}
     \end{equation}
     The link SINR is lower bounded by
${\ell(\tau)}/{(A_\ell-\ell(\tau)+W)},~\mathbb{P}_{\Psi}^{o}\mhyphen \mathrm{a.s.}$, and the ergodic rate is lower bounded as by $ \log\left(1+{\ell(\tau)}/{(A_\ell-\ell(\tau)+W)}\right),
~\mathbb{P}_{\Psi}^{o}\mhyphen \mathrm{a.s.}\nonumber$
    \end{corollary}
\begin{proof}Since $\Phi$ is $(\sigma,\rho,\nu)$-ball regulated with respect to $\Psi$, from Definition \ref{thm: observer-G-ball-sn-eq}, $\sum_{x\in \Phi}\ell(\|x\|)\leq A_\ell,\mathbb{P}_{\Psi}^{o}\mhyphen \mathrm{a.s.}$ Then it follows from the fact that $I_t = \sum_{x\in \Phi\setminus\{t\}} \ell(\|x\|) =\sum_{x\in \Phi}\ell(\|x\|) -\ell(\tau)$. 
\end{proof}
\begin{remark}
Without spatial regulation on $\Phi$, e.g., when $\Phi$ is a PPP, the above a.s. bounds do not exist.
\end{remark}
\begin{example}
    Consider $\ell(r) = \min\{1,r^{-\alpha}\},~\alpha>2$. If $\Phi$ strongly is $(\sigma,\rho,\nu)$-ball regulated, by the strong-weak relation,
\[    I_t\leq \sigma +{\rho\alpha}/{(\alpha-1)}+{\nu\alpha}/{(\alpha-2)}-\min\{1,\tau^{-\alpha}\}, ~\mathbb{P}_{\Psi}^{o}\mhyphen \mathrm{a.s.}\]
 For a hardcore point process with hardcore distance $H=1$, $\alpha=4$ and $\tau=1$, one can simply plug into  $(\sigma,\rho,\nu)=(1,2\pi/\sqrt{12},\pi/\sqrt{12}))$ from Lemma \ref{lemma: hc}.
 \end{example}


\subsection{With Fading}
\subsubsection{Interference bounds}
With fading, even if $\Phi$ is ball regulated w.r.t. $\Psi$, the power of total interference given $\Phi$ cannot be a.s. upper bounded in general due to potentially unbounded support of fading. However, its tail distribution can be bounded using classical inequalities.
\begin{corollary}[Markov and Chebyshev Bounds]
\label{corr: markov-chebyshev-I}
If $\Phi$ is $(\sigma,\rho,\nu)$-ball regulated w.r.t. $\Psi$, then $\mathbb{P}_{\Psi}^{o}\mhyphen \mathrm{a.s.}$, 
$\mathbb{P}_{\Psi}^o (I>x\mid\Phi)\leq {(A_{\ell}-\ell(\tau))}/{x}$, and 
    $\Pr_{\Psi}^o(|I-\Ex [I\mid \Phi]|>x\mid\Phi)\leq {(\Ex  h^2-1) \left(A_{\Tilde\ell}-\Tilde\ell(\tau)\right)}/{x^2}$ where $\Tilde\ell=\ell^2.$
 \end{corollary}
\begin{proof}
The first inequality follows from the fact that $\Ex [I\mid \Phi]=\sum_{x\in\Phi\setminus\{t\}}\ell(\|x\|)\leq A_\ell-\ell(\tau)$ from Corollary \ref{corr: D2D_wo_fading} and the Markov inequality. 
For the second, the conditional variance of $I$ given $\Phi$ is
\begin{align}
&\Ex_{\Psi}^o\left[\left(\sum_{x\in\Phi\setminus\{t\}}(h_x-1)\ell(\|x\|)\right)^2~\Bigg|~\Phi\right]\nonumber\\
&\peq{a}\sum_{x\in\Phi\setminus\{t\}}\Ex \left[(h_x-1)^2\right]\ell^2(\|x\|)\nonumber\\& = \left(\Ex [h^2]-1\right)\sum_{x\in\Phi\setminus\{t\}}\tilde\ell(\|x\|).\nonumber
   \nonumber\label{eq: var-L}
\end{align}
Step (a) follows from the independence assumption of fading and $\Ex h =1$. Then it follows from Chebyshev inequality and the fact that $\ell^2$ is bounded and non-decreasing, and that $\Phi$ is $(\sigma,\rho,\nu)$-shot-noise regulated w.r.t. $\Psi$. 
\end{proof}
 \begin{corollary}
\label{corr: int-mgf}
For the D2D bipolar model described in Section \ref{sec: d2d}-A, if $\Phi$ is $(\sigma,\rho,\nu)$-ball regulated w.r.t. $\Psi$,  then there $\exists 
 s^*>0$ such that for $s\in[0,s^{*}),$
$  \mathcal{L}_{I\mid\Phi}(-s)\leq \exp\left(A_{\Tilde\ell}-\tilde\ell(\tau)\right),~ \mathbb{P}_{\Psi}^{o}\mhyphen \mathrm{a.s.},$ where ${\Tilde\ell}(r)= \log\mathcal{L}_h(-s\ell(r))$.
 \end{corollary}
 \begin{proof}
Follows from Theorem \ref{thm: sn-laplace}.
\end{proof}
If $\Phi$ is not ball-regulated, $\Ex [\exp(sI)\mid \Phi]$ cannot be $\mathbb{P}_{\Psi}^{o}\mhyphen$a.s. upper bounded. Hence the Laplace transform of the interference with fading cannot be $\mathbb{P}_{\Psi}^{o}\mhyphen$a.s. upper bounded. One such an example is the PPP.
\begin{example}
    For i.i.d. Nakagami-$m$ fading, $\mathcal{L}_h(-s) = (1-s/m)^{-m}$, and $s^{*}=m$. When $\ell(r) = \min\{1,r^{-\alpha}\}$, $\alpha>2$, 
\begin{equation}
    \begin{split}
    & A_{\Tilde\ell}=m\log\frac{1}{(1-s/m)^\sigma}+ \frac{\alpha\rho s~{_2F_1}\left(1-\frac{1}{\alpha},1;2-\frac{1}{\alpha};\frac{s}{m}\right)}{\alpha-1}\\
&\quad\quad+\frac{\alpha \nu s~{_2F_1}\left(1-\frac{2}{\alpha},1;2-\frac{2}{\alpha};\frac{s}{m}\right)}{\alpha-2},\nonumber
\end{split}
\end{equation} 
where $_2F_1$ is the hypergeometric function and can be efficiently evaluated numerically. 
\end{example}


\begin{corollary}[Chernoff bound]
Let ${\Tilde\ell}(r) = \log\mathcal{L}_h(-s\ell(r))$.
If $\Phi$ is $(\sigma,\rho,\nu)$-ball regulated, then
    \begin{align}
    \Pr_{\Psi}^o(I>x\mid\Phi)\leq \exp{\left( \inf_{s\in[ 0,s^{*})}-sx+A_{\Tilde\ell}-\tilde\ell(\tau)\right)},~\mathbb{P}_{\Psi}^{o}\mhyphen \mathrm{a.s.}
    \label{eq: exp_min}
\end{align}
\label{corr: Chernoff}
\end{corollary}
\begin{proof}
For $\forall s\geq0$,
$ \Pr_{\Psi}^o(I>x\mid\Phi)=\Pr_{\Psi}^o\left(\exp\left(sI\right)>\exp(sx)\mid\Phi\right)\leq \exp(-sx)\mathcal{L}_{I\mid\Phi}(-s)$. Then Eq (\ref{eq: exp_min}) follows from Corollary \ref{corr: int-mgf} and minimizing the exponent over $s$.\end{proof}
The exponent in the RHS of (\ref{eq: exp_min}) is known as the Legendre transform. If it is differentiable in $s$, one can solve it by taking the derivative.
\subsubsection{Performance guarantees}

\begin{theorem}[Reliability lower bound for all links]
If $\Phi$ is $(\sigma,\rho,\nu)$-ball regulated with respect to $\Psi$ and the fading satisfies the exponential moment condition, then for any $\theta>0$, the link reliability is lower bounded. Moreover,
$P_{\mathrm{s}}(\theta)\geq \zeta(\theta),~\mathbb{P}_{\Psi}^{o}\mhyphen \mathrm{a.s.},$
where $\zeta(\theta)$ is the following deterministic function:
\begin{equation}
\zeta(\theta)\triangleq \int_{\frac{W\theta}{\ell(\tau)}}^{\infty} f_{h}(x)\left(1-e^{\inf_{s\in[0,s^*)}A_{\Tilde\ell}-\tilde\ell(\tau)-s\left(x\ell(\tau)\theta^{-1}-W\right)}\right)^+\dd x,
\label{eq: lb-zeta}
\end{equation}
and $\Tilde\ell(r) = \log\mathcal{L}_h(-s\ell(r)).$ Further, for any $\epsilon>0$ , $\exists \delta>0$ such that a reliability $1-\epsilon$ is achieved by all links for $\theta\leq \delta$.
\label{thm: csp-lb}
\end{theorem}
\begin{proof}
See Appendix \ref{appendix: thm-bipolar-P_m}.
\end{proof}
\begin{remark}[Generality]
    Theorem \ref{thm: csp-lb} is applicable to all $(\sigma,\rho,\nu)$-ball regulated point processes (strong or w.r.t. $\Psi$) regardless of their distributions. For channel models, it applies to all bounded and non-increasing path loss functions and fading statistics including Rayleigh, Rician, and Nakagami.
\end{remark}
  Theorem \ref{thm: csp-lb} proves lower-bounded ultra-reliability for all links.  $\zeta(\theta)$ can be efficiently computed using numerical methods.  With Nakagami fading and $\ell(r)= \min\{1,r^{-\alpha}\}$, one can find the infimum of the exponent by taking its derivative.

\begin{theorem}
\label{thm: bipolar-Rayleigh-Pm}  For Rayleigh fading, if $\Phi$ is $(\sigma,\rho,\nu)$-ball regulated w.r.t. $\Psi$, then
\begin{equation}
\label{eq: bipolar-Rayleigh-Pm}
 P_{\mathrm{s}}(\theta)\geq  {\exp\left(-\frac{\theta W}{\ell(\tau)}\right)} {\exp\left(-A_{\Tilde\ell}+\tilde\ell(\tau)\right)},
 \quad\mathbb{P}_{\Psi}^{o}\mhyphen \mathrm{a.s.},
\end{equation}
where $\Tilde\ell(r) = \log(1+\theta\ell(r)/\ell(\tau))$.
\end{theorem}
\begin{proof}
\begin{align}
     P_{\mathrm{s}}(\theta)&\peq{a}\Ex_{\Psi}^o\left[\exp\left(-\frac{\theta\left(I+W\right)}{\ell(\tau)}\right)~\bigg|~ \Phi\right]\nonumber\\
    & \peq{b} {\exp\left(-\frac{\theta W}{\ell(\tau)}\right)} \prod_{x\in\Phi\setminus\{t\}}\frac{1}{1+\theta\ell(\|x\|)/\ell(\tau)}\nonumber\\
     &= {\exp\left(-\frac{\theta W}{\ell(\tau)}\right)} {\exp\left(-\sum_{x\in\Phi\setminus\{t\}}\log(1+\theta\ell(\|x\|)/\ell(\tau))\right)}.\nonumber
\end{align}  
Step (a) follows from the CCDF of Rayleigh fading. Step (b) follows from the independence assumption on fading and the distribution of Rayleigh fading. Lastly, $\Tilde\ell(r) = \log(1+\theta\ell(r)/\ell(\tau))$ is non-increasing, monotonic and bounded. Hence we prove the theorem by applying the upper bound for the interference for a $(\sigma,\rho,\nu)$-regulated point process w.r.t. $\Psi$. 
\end{proof}

\begin{figure}
    \centering
    \begin{subfigure}[$\alpha=3.$]{
    \includegraphics[width=.9\linewidth]{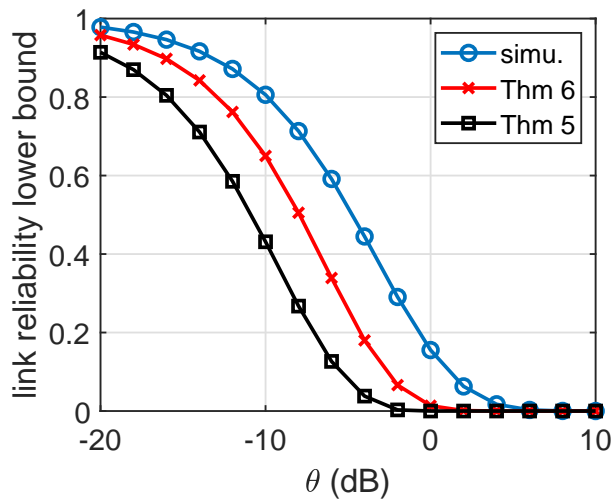}
  }
  \end{subfigure}
      \begin{subfigure}[$\alpha=4.$]{
    \includegraphics[width=.9\linewidth]{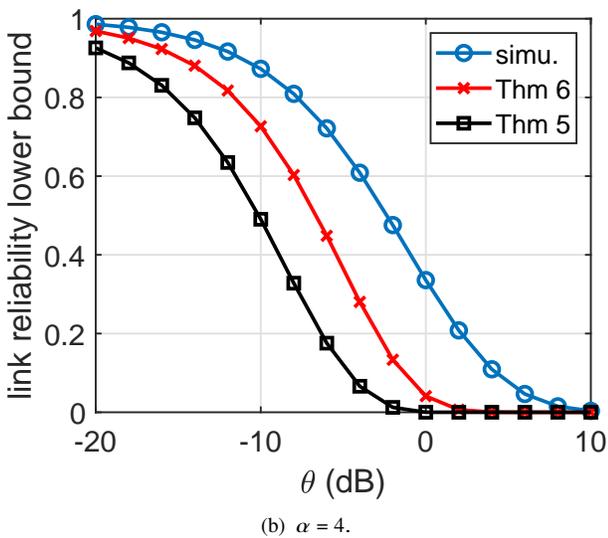}
      }
     \end{subfigure}
    \caption{Lower bounds of link reliability in a bipolar network where $\Phi$ is a triangular lattice with $H=1$. $W=0,~\ell(r)=\min\{1,r^{-\alpha}\}$, $\tau=1$, Rayleigh fading. The analytical lower bounds are obtained via Theorems \ref{thm: bipolar-Rayleigh-Pm} and \ref{thm: csp-lb}, with $A_{\tilde \ell}-\tilde\ell(\tau)$ given in (\ref{eq: A_l-bipolar-example}) and $\sigma=1, \rho=2\pi/\sqrt{12},\nu=\pi/\sqrt{12}$ from Lemma \ref{lemma: hc}. The simulated lower bounds are obtained via 500,000 realizations.}
    \label{fig: bipolar-Rayleigh-Pm}
\end{figure}

Theorem \ref{thm: bipolar-Rayleigh-Pm} gives an explicit formula for an a.s. lower bound  on $P_{\rm{s}}(\theta)$ in the special case of Rayleigh fading. By the strict monotonicity and continuity w.r.t. $\theta$ of the RHS expression in (\ref{eq: bipolar-Rayleigh-Pm}), we can invert the RHS of (\ref{eq: bipolar-Rayleigh-Pm}) to obtain a lower bound on the threshold SINR needed to achieve an arbitrarily high target reliability. 
{Specially, for $\ell(r) = \min\{1,r^{-\alpha}\}$,}
\begin{align}
A_{\Tilde\ell}&=  \log\frac{\left(1+\frac{\theta}{\ell(\tau)}\right)^\sigma}{1+\theta}+\frac{\alpha\rho \theta~{_2F_1}\left(1-\frac{1}{\alpha},1;2-\frac{1}{\alpha};-\frac{\theta}{\ell(\tau)}\right)}{(\alpha-1)\ell(\tau)}\nonumber\\
&\quad+\frac{\alpha \nu \theta~{_2F_1}\left(1-\frac{2}{\alpha},1;2-\frac{2}{\alpha};{-\frac{\theta}{\ell(\tau)}}\right)}{(\alpha-2)\ell(\tau)}+\tilde\ell(\tau).\label{eq: A_l-bipolar-example}
\end{align}

In Fig. \ref{fig: bipolar-Rayleigh-Pm}, we plot the derived lower bounds against simulated lower bounds where $\Phi$ is a stationary triangular lattice with hardcore distance is 1, and $\ell(r)=\min\{1,r^{-\alpha}\}$ with $\alpha=3,~4$ respectively. Here we choose triangular lattice to obtain a relatively accurate estimation of the worst-case reliability. While the same a.s. bound holds for a Mat\'ern hardcore point process with the same hardcore distance, extreme cases happen rarely and hence are more difficult to simulate accurately. Analysis on statistical bounds are left for future work. The $(\sigma,\rho,\nu)$-ball regulation parameters are given in Lemma \ref{lemma: hc}. The vertical gap is the difference in the worst-case link reliability between the proposed analytical characterization versus simulations. The horizontal gap is the difference between the derived and simulated SINR thresholds needed to achieve a given target reliability. For a target reliability above 0.7, there is a horizontal gap less than 4 dB between simulation and Theorem \ref{thm: bipolar-Rayleigh-Pm}. The bounds for $\alpha=3$ measured by the horizontal gap are slightly tighter than the bounds for $\alpha=4$.

\begin{figure}
    \centering
    \begin{subfigure}[$\alpha=3.$]{
    \includegraphics[width=.9\linewidth]{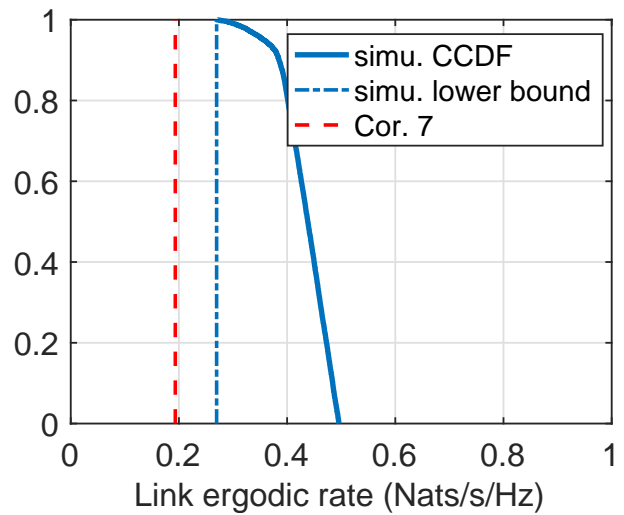}
  }
  \end{subfigure}
      \begin{subfigure}[$\alpha=4.$]{
    \includegraphics[width=.9\linewidth]{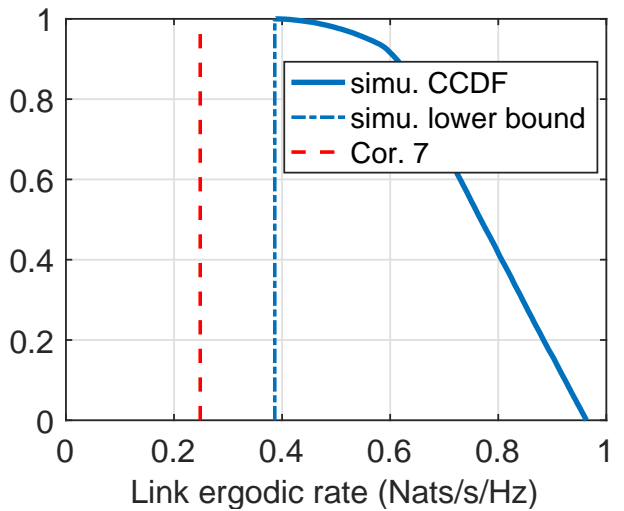}
      }
     \end{subfigure}
    \caption{Lower bounds and CCDF of the link ergodic rate in a bipolar network. The analytical lower bounds are obtained via Corollary \ref{cor: ergodic} (using the bound in Theorem \ref{thm: bipolar-Rayleigh-Pm}). The setup and parameters are the same as those for Fig. \ref{fig: bipolar-Rayleigh-Pm}.}
    \label{fig: bipolar-Rayleigh-C}
\end{figure}

\begin{corollary}[Ergodic rate  lower bound for all links]
\label{cor: ergodic}
  If $\Phi$ is $(\sigma,\rho,\nu)$-ball regulated with respect to $\Psi$, then the ergodic rate is lower bounded by $\int_{0}^{\infty} \exp{\left(-\zeta(e^{t}-1)\right)}\dd t,
\mathbb{P}_{\Psi}^{o}\mhyphen \mathrm{a.s.},$
where $\zeta(\cdot)$ is defined in Theorem \ref{thm: csp-lb}. 
	Let $\Tilde\ell(r)\triangleq \log(1+(e^{t}-1)~\ell(r)/\ell(\tau))$.
	For Rayleigh fading, we have 
   \begin{equation}
\begin{split}
    C\geq \int_{0}^{\infty} {\exp\left(-\frac{(e^{t}-1) W}{\ell(\tau)}\right)} {\exp\left(-A_{\Tilde\ell}+\tilde\ell(\tau)\right)}\dd t,
~\mathbb{P}_{\Psi}^{o}\mhyphen \mathrm{a.s.}
\end{split}
\end{equation}
\end{corollary}
\begin{proof}
With fading,
\begin{align}
C&= \int_{0}^{\infty} \Pr_{\Psi}^o (\log (1+\mathrm{SINR})>t\mid \Phi) \dd t
\nonumber\\
&=\int_{0}^{\infty} \Pr_{\Psi}^o ( \mathrm{SINR}>e^{t}-1\mid\Phi)\dd t.
\label{eq: rate-eq}
\end{align}
The rest then follows from the reliability lower bounds in Theorems \ref{thm: csp-lb} and \ref{thm: bipolar-Rayleigh-Pm}.
\end{proof}

Fig. \ref{fig: bipolar-Rayleigh-C} plots the analytical lower bound in Corollary \ref{cor: ergodic} (using the bound in Theorem \ref{thm: bipolar-Rayleigh-Pm}) against simulated lower bounds for the link ergodic rate, where the parameters are the same as those in Fig. \ref{fig: bipolar-Rayleigh-Pm}. The gaps between the simulated lower bound and the analytical lower bound for $\alpha=3$ and $\alpha=4$ are 0.07 and 0.148 (in Nats/s/Hz), respectively. 

\subsection{Wireless Queues in D2D Networks}
Spatial network calculus, when combined with classical network calculus, induces performance guarantees for all queues simultaneously in wireless queueing systems. Previously, stochastic network calculus has been applied to study performance of wireless queues by formulating SINR based service curves \cite{al-zubaidy-13infocom,kountouris18QoS}. For large wireless networks, \cite{kountouris18QoS} analyzes and latency of wireless queues in Poisson bipolar networks. Due to the absence of spatial regulations in Poisson networks, stability and performance guarantees only hold for a (strict) subset of links in the network \cite{yi16stability}.

Consider a countable collection of queues associated with the links, where each queue has an infinite buffer with first-in-first-out (FIFO) scheduling. Traffic arrives at the transmitter according to a time stationary and ergodic point process with intensity $\lambda>0$, is served by the wireless link,
and leaves the queue when successfully decoded by the receiver.
The service process at each queue depends on the performance of its wireless link. 
We assume full buffer at interfering queues, which is conservative.
Another essential matter concerns fading now seen as a time series.
Let $h_x(t)$ denote the fade from $x\in \Phi$ to the origin at time $t$.
Conditionally on $\Phi$ the time fading processes $\{h_x(t)\}_t$ are assumed
i.i.d. w.r.t. $x\in \Phi$ with exponential moments, which is in line with what was assumed above.
The main additional assumption is that for all fixed $x$, the time series $h_x(t)$
is ergodic. A simple instance is  where this time series is a mixing Markov chain,
with a steady state distribution which is that considered above (say Nakagami).
\begin{theorem}
    Let $\Phi $ be a stationary and ergodic spatial point process denoting the transmitter locations in a bipolar network.
    If $\Phi$ is $(\sigma,\rho,\nu)$-ball regulated with respect to $\Psi$, then
    \begin{enumerate}
    \item in the absence of fading, if the traffic arrival to each queue is $(\sigma^{*},\rho^*)$-regulated, with $\rho^*<\log\left(1+\frac{\ell(\tau)}{A_\ell-\ell(\tau)+W}\right) $, then the queue length is upper bounded, and the latency is upper bounded at each queue;
     \item in the presence of ergodic fading with exponential moments, there exists a threshold $\rho^*>0$ such that when $\lambda<\rho^*$, all queues in the network are stable, the tail distribution of latency is upper bounded, and the tail distribution of queue length is upper bounded.
    \end{enumerate}
    If $\Phi$ is not ball-regulated with respect to $\Psi$, either in the absence of fading or with fading, there is a non-zero fraction of unstable queues. The tail distribution of link-level latency and queue length cannot be upper bounded.
    \end{theorem} 
The proof is omitted. The statement on the existence of bounds on the tail distribution of latency follows from Loynes' theorem \cite{bacelli1994elements} and leverages the time ergodicity assumptions on the fading time series and the arrival point processes. 
The nature of the result
is illustrated by the following discrete-time example.


\begin{example}[Stochastic bounds on latency in discrete-time queues]
The aim of this example is to illustrate the stochastic bounds that can be obtained
on queuing latency by spatial regulation alone, i.e., in the absence of time
regulation and in the presence of fading with short coherence times.
The illustration features a D2D network with Rayleigh fading where
each link is equipped with an infinite buffer FIFO queue.
Time is slotted. The packet arrival process in each queue is i.i.d. over time slots and the number of
arrivals of packets in a time slot has the distribution $X$ over the integers.
The packet head of a given queue at a given time slot is successfully transmitted (served) if  
the SINR at that location and at that time exceeds the threshold $\theta$.
If it is the case, the packet leaves. Otherwise it stays head of the line and tries again next slot
until its transmission succeeds. 
The simplest scenario is that where the slot duration is of the order of the coherence
time of the fades. In this case, in each queue, the fades experienced 
in different time slots are i.i.d.  
In the Rayleigh fading case, the conditional probability of success in the typical queue given the transmitter
point process is bounded from below by the constant given in Theorem \ref{thm: csp-lb}.
Hence, in each queue, the stationary latency is stochastically bounded from above
by that in a discrete time queue with i.i.d. arrivals with distribution $X$ and i.i.d.
Bernoulli service process with a constant probability of success $\zeta(\theta)$. 
The distribution of the stationary queue size and of the latency in such queues has
been extensively studied using generating function techniques and is known in closed form
in function of $X$ and $\zeta(\theta)$ \cite{bruneel1993performance}.
\end{example}
The results of the last example can be extended to more complex scenarios 
(e.g., a mixing Markov evolution
of the fades over time slots in place of i.i.d. assumptions).
There is a wealth of well-known computational results on the tail decay of stationary queues\cite{asmussen2003applied},
starting with the seminal work of W. Feller for GI/GI/1 queues, covering both
discrete-time queues and fluid queues, which we will not review here
despite the fact that they are fully relevant in the wireless setting considered in the last theorem.

The ergodicity of the small-scale fading is a necessary condition for the stability of all queues. Without this ergodicity, a wireless link may stay with an arbitrarily low rate  due to persistent bad fades, which leads to a situation where a positive 
fraction of the queues are unstable.

\section{Performance Guarantees In Ad hoc Networks}
\label{sec: ad-hoc}
\subsection{System Model}
Let $\Phi$ be a stationary point process on $\mathbb{R}^2$ modeling the locations of transceivers, defined on a probability space $(\Omega,\mathcal{A},\mathbb{P})$. By the strong-weak regulation relation in Section \ref{subsec: strong-weak}, if $\Phi$  is strongly ball (void) regulated, then it is also weakly ball (void) regulated with respect to itself.

We assume that node $x$ can transmit to node $y$ if the SINR received by $y$  is large enough. As above, $\ell: \mathbb{R}^+\to\mathbb{R}^+$ denotes the path loss function, which depends only on the link length, is bounded, non-increasing, and integrable in $\mathbb{R}^2$. $W$ denotes the variance of the additive white Gaussian noise. Here, we denote the fading from node $x$ to $y$ by $h_{xy}$, and assume that the fading random variables $\{h_{xy}\}_{x,y\in\Phi,x\neq y}$ are i.i.d. with exponential moments.
The SINR of the link $x\to y$, measured at node $y$ is 
\[ \mathrm{SINR}_{xy} \triangleq \frac{{h_{xy}\ell(\|x-y\|)}}{I+W},\]
where
$I=I_{x,y}\triangleq\sum_{z\in \Phi\setminus\{x,y\}} h_z \ell(\|z-y\|)$ is the power of the total  received interference at $y$.

The SINR graph of $\Phi$ for a predefined SINR threshold $\theta$ is the random graph with nodes the atoms of $\Phi$ and with an edge between $x$ and $y$ in $\Phi$ if $\mathrm{SINR}_{xy}>\theta$ and $\mathrm{SINR}_{yx}>\theta$ \cite{dousse05connectivity, baccelli2010stochastic}.
One says that the SINR graph of $\Phi$ {\em percolates} for the threshold $\theta$ if it has an infinite component. One says that it has a north-south and east-west (or any other direction) {\em backbone} if,
under the Palm probability of $\Phi$, there exists a $\tau>0$ and points 
$\{x_{i,j}\}_{i,j\in \mathbb Z}$ of $\Phi$ such that \begin{enumerate}
    \item $x_{0,0}=(0,0)$ and for all $i\in \mathbb{Z}$,
$x_{i,j}$ belongs to the square of center $(2i\tau,2j\tau)$ and side $2\tau$;
\item there is an edge between $x_{i,j}$ and $x_{i{+\atop -}1,j}$ 
and an edge between  $x_{i,j}$ and $x_{i,j{+\atop -}1}$, 
for all $i,j\in \mathbb Z$.
\end{enumerate}
Note that the existence of such a SINR backbone in the SINR graph implies its percolation. 


\subsection{SINR Graph Percolation}
If $\Phi$ is strongly $(\sigma,\rho,\nu)$-ball regulated and strongly $\tau$-void regulated, then
\begin{enumerate}
    \item   in the absence of fading, the interference at the origin is $\mathbb{P}_{\Phi}^{o}\mhyphen \mathrm{a.s.}$ bounded above by $A_\ell-\ell(\tau)$. From the nearest node to $o$, the link rate when treating interference as noise is lower bounded by $\log\left(1+\frac{\ell(\tau)}{A_\ell-\ell(\tau)+W}\right),~\mathbb{P}_{\Phi}^{o}\mhyphen \mathrm{a.s.}$;
\item If fading has exponential moments, the conditional Laplace transform of interference is $\mathbb{P}_{\Phi}^{o}\mhyphen \mathrm{a.s.}$ bounded above by $\exp(A_{\tilde\ell}-\tilde\ell(\tau))$, where $\Tilde\ell(r) = \log\mathcal{L}_h(-s\ell(r))$. From the nearest node to $o$, the link reliability is lower bounded by $\zeta(\theta)$, $\mathbb{P}_{\Phi}^{o}\mhyphen \mathrm{a.s.}$ The  ergodic rate is lower bounded by $\int_{0}^{\infty} \exp{\left(-\zeta(e^{t}-1)\right)}\dd t,~\mathbb{P}_{\Phi}^{o}\mhyphen \mathrm{a.s.}
$
\end{enumerate}

\begin{lemma}
If $\Phi$ is strongly $(\sigma,\rho,\nu)$-ball  regulated and strongly $\tau$-void regulated, then
\begin{enumerate}
    \item in the absence of fading, the SINR graph admits a north-south and east-west backbone for $\theta$ below a positive threshold;
    \item in the presence of fading, the SINR graph percolates for $\theta$ below a positive threshold.
\end{enumerate}
\end{lemma}
\begin{proof}
Let us first build the backbone.
Since $\Phi$ is $\tau$-void regulated, with probability 1, 
each ball $B((2i\tau,2j\tau)), \tau)$ contains at least one point of $\Phi$.
For each $(i,j)\ne (0,0)$, let $x_{i,j}$ be the point of that ball with, e.g., the largest 
abscissa. We now show that for $\theta$ small enough, in the absence of fading, the sequence $\{x_{i,j}\}$ is a $2\tau$ north-south and east-west backbone.
Two points $x=x_{i,j}$ and $y=x_{l,k}$ of this backbone will be said to be neighboring points if 
$(i,j)$ and $(l,k)$ are neighbors in the $\mathbb{Z}^2$ grid.
In the absence of fading, the interference at any point of this backbone is upper bounded by $A_\ell$, $\mathbb{P}_\Phi^o$-a.s. and hence ${\rm{SINR}}_{xy}$ is lower bounded.  So all neighboring points are connected for $\theta$ small enough, which concludes the proof of the first statement.

For the second part of the lemma, note that, by an argument similar
to that used in the proof of Theorem \ref{thm: csp-lb},
for every $\epsilon>0$, there exists a small enough $\delta>0$ such that, e.g., $\Pr^o_\Phi ({\rm{SINR}_{x_{1,0} x_{0,0}}}>\delta)\geq 1-\epsilon,$ $\mathbb{P}_\Phi^o$-a.s.
Hence there exists  a small enough SINR threshold such that any two neighboring points of
the backbone are connected with probability more than 1/2. Using now the fact that 
the problem of percolation of the backbone
is isomorphic to that of bond percolation on $\mathbb{Z}^2$ with bonds
being independently closed, one can leverage the bond percolation theory \cite{grimmett1999percolation} to show that the SINR graph percolates when this
probability is more than 1/2.
\end{proof}

\subsection{Ad hoc Queueing Networks}
An important application of network calculus is in queues in series, where the concatenation of a set of service curves can be directly calculated by techniques similar to those in system theory \cite{chang2000performance, le2001network}, but in the min-plus algebra. 
It makes sense to leverage the percolation properties discussed in the last subsection to
build networks of queues that allow for long distance multi-hop relaying in this ad hoc setting. For instance, in the case without fading, the backbone structure can be leveraged to
maintain networks made of queues in series on each east-west path and each north-south path with a guaranteed service curve (Shannon rate) in each station (link between two adjacent transceivers). By classical queueing theory arguments, one can maintain long distance flows between any pairs of points of the backbone at a positive rate (independent of the distance between source and destination) by multi-hop routing leveraging only transmissions from one 
queue to the next on the backbone.
If the input flows are time regulated, in the absence of fading, one can even provide deterministic end-to-end latency guarantees on these flows by adapting
the network calculus theory developed for queues in series (see \cite{chang2000performance, le2001network}) to the wireless links in series used in this multi-hop routing scheme.
\section{Performance Guarantees in Cellular Networks}

\label{sec: cellular}
\subsection{System Model}
\label{subsec: system}
Let the base station (BS) and user locations in a cellular network be modeled by jointly stationary and ergodic point processes $\Phi$ and $\Psi$ on $\mathbb{R}^2$, defined on a probability space $(\Omega,\mathcal{A},\mathbb{P})$.  Let each user be associated with its strongest (in average) BS, which, by the monotonicity of the path loss function, is the nearest BS. The nearest BS is unique almost surely \cite{baccelli:hal-02460214}. By this association scheme, each BS serves users within its Voronoi cell via multiple accesses,  and so the intra-cell resources is shared among users in time and frequency. This feature mandates what we call cell-load regulation. Denote by $V$ the Voronoi cell associated with the typical BS $o$.

\begin{definition}[Cell-load regulation]
A stationary cellular network consisting of a BS point process $\Phi$ and a user point process $\Psi$ is cell-load regulated if there exists some constant $K>0$ such that $
\Psi{({V})}\leq K, \Pr_{\Phi}^o\mhyphen \mathrm{a.s.}
$
\end{definition}
In essence, cell-load regulation controls
the clustering of $\Psi$ in the Voronoi cells of $\Phi$, such that the effect of intra-cell multiple access is accounted for in the dowlink performance. Cell-load regulation can be implemented cell by cell through a
centralized trimming implemented in each base station, in a
decentralized way, or through load balancing. For instance, if the cell surface is shared, then an appropriate ball regulation on users leads to a distributed
implementation.
\begin{lemma}
A cellular model as described above is cell-load regulated if $\Phi$ is strongly void regulated and $\Psi$ is ball regulated w.r.t. $\Phi$.
\end{lemma}
\begin{proof}
Let $\Phi$ be strongly $\tau$-void regulated. The Voronoi cell associated with the origin $V\subset B(o,\tau)~ \Pr_{\Phi}^o\mhyphen\mathrm{a.s.}$ If not, by the convexity of Voronoi cells, there exists a non-zero fraction of $V$ such that $B(\cdot,\tau)=0,\Pr_{\Phi}^o\mhyphen \mathrm{a.s.}$, which contradicts the assumption of strong $\tau$-void regulation. Since $\Psi$ is ball regulated w.r.t. $\Phi$, then we must have  for some constant $K>0$, $\Pr_{\Phi}^o (\Psi(B(o,\tau))\leq K)=1$. Hence $\Pr_{\Phi}^o(\Psi{({V})}\leq K)=1$.
\end{proof}

\begin{figure}[t]
    \centering
\includegraphics[width=0.85\linewidth]{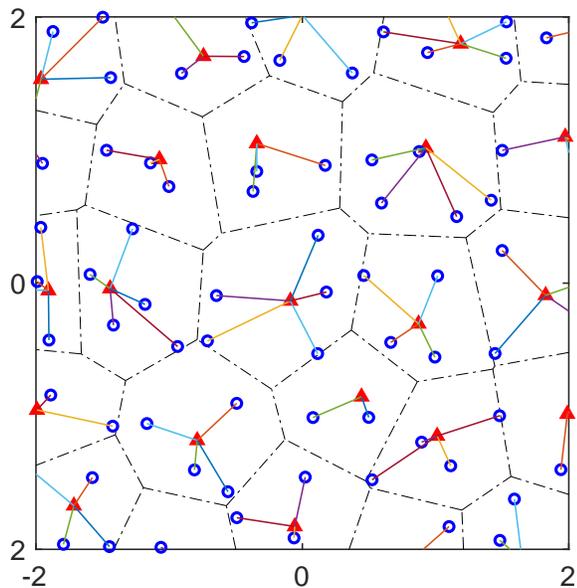}
    \caption{Illustration of a spatially and cell-load regulated cellular network, where BSs are indicated by red triangles and users by blue circles. BSs and users are realizations of two independent i.i.d. perturbed square lattices, with $\lambda_{\Psi}=4\lambda_{\Phi}=1$. Dashed black lines indicate Voronoi cell boundaries of the base station point process. }
    \label{fig:cellular}
\end{figure}

\subsection{Link Performance Guarantees}
We focus on the typical user at the origin.  Let $t\triangleq\argmax\{\ell(\|x\|): x\in\Phi\}$ be the serving BS of the typical user. The interference at the typical user is
$I =I_t \triangleq \sum_{x\in\Phi\setminus\{t\}}h_x\ell(\|x\|).$
For the cellular model described in Section 
\ref{subsec: system},
if $\Phi$ is $\tau$-void and ball regulated w.r.t. $\Psi$, then
\begin{enumerate}
    \item 
in the absence of fading, $\mathbb{P}_{\Psi}^{o}\mhyphen \mathrm{a.s.}$, $I \leq A_\ell-\ell(\tau)$, and $\mathrm{SINR} \geq {\ell(\tau)}/{(A_\ell-\ell(\tau)+W)}$;
    \item If fading has exponential moments, then  $\mathbb{P}_{\Psi}^{o}\mhyphen \mathrm{a.s.}$, $\mathcal{L}_{I\mid\Phi}(-s)\leq \exp\left(A_{\Tilde\ell}-\tilde\ell(\tau)\right)$ where ${\Tilde\ell}(r)= \log\mathcal{L}_h(-s\ell(r))$, and the link reliability is lower bounded by $\zeta(\theta)$.
\item if in addition, the network is cell-load regulated, then the information-theoretic rate of each user under shared BS time-frequency resources can be lower bounded.
\end{enumerate}

It is easy to apply the superposition properties of regulations to derive similar bounds for multi-tier networks. 

Fig. \ref{fig:cellular} shows a spatially and cell-load regulated cellular network where both BS and users follow two independent i.i.d. perturbed square lattices, with $\lambda_{\Psi}=4\lambda_{\Phi}=1$. The displacement vectors are with length uniform randomly distributed in $[0,0.5]$ and $[0,0.25]$, respectively, and angles are uniform randomly distributed in $[0,2\pi]$.
\subsection{Wireless Queues in Cellular Networks}

Consider a countable collection of infinite buffer FIFO queues associated with BS-user pairs. Assume that each BS serves its associated users using multiple access schemes. Assume that traffic arrives at the queues of each BS according to a time stationary and ergodic point process with intensity $\lambda>0$, and leaves the queue when successfully decoded at the intended user. Let the cellular network be cell-load regulated and $\Phi$ be $(\sigma,\rho,\nu)$-ball regulated w.r.t. $\Psi$.   Consider the queue associated with the BS of the typical user and assume full buffer at interfering queues.
  In the absence of fading, $\exists \rho^*>0$ such that if the downlink traffic arrivals at each BS-user queue is $(\sigma^*,\rho^*)$-regulated,  the downlink  latency is upper bounded;
    with ergodic fading having exponential moments, if the downlink traffic arrival process at each BS-user queue is ergodic, then there exists an average rate such that the tail distribution of the downlink latency is upper bounded.

 \section{Discussion}

 In this work, we derive computable bounds on performance  for all links in an arbitrarily large wireless network under appropriate spatial regulations.  The proposed spatial regulations can be algorithmically implemented.
The strong void regulation can be implemented by, e.g., adding a regular or
perturbed grid. 
The ball regulation can be implemented in a distributed way by the Matérn I or II method \cite{baccelli2010stochastic}.
 It is beyond the scope of the present paper to discuss the
distributed implementation of weak regulation. A critical question is about the achievability of these bounds. In this direction, we made the first steps by defining the extremal parameters relating to $(\sigma,\rho,\nu)$-ball regulation and introducing the generalized ball regulation. Of particular interests are tight performance bounds in targeted use cases, e.g., ultra-reliability, low latency, or high throughput. Wireless networks can support multiple classes of users/traffic with different priorities, whose joint modeling and optimization is worth exploring. Other relevant problems include the impact of shadowing, dynamic power control, load balancing, and directional transmission on performance guarantees.

\appendix


 
\subsection{Proof for Lemma \ref{lemma: hc}}
\label{appendix: hardcore}
Consider any hardcore point process $\Phi$ on $\mathbb{R}^2$ with hardcore distance $H>0$. For any $r>0$, draw balls of radius $H$ centered at points in $\Phi\cap b(o,r)$. The union of the balls $\cup_{x\in \Phi\cap b(o,r)} b(x,H) \subset b(o,r+H)$. The number of points of $\Phi$ in $b(o,r)$ must be bounded from above by the maximum number of balls of radius $H$ one can pack in $b(o,r+H)$ (known as densest packing). Further, the following results are known about densest packings\footnote{In $\mathbb{R}^2$, the problem is  known as circle packing. Here, we use the terms ``balls'' and for packing problems in general dimensions. }: the volume fraction\footnote{This quantity is usually referred to as the ``packing density'' in sphere packing problems. Here, to avoid confusion with the spatial density of points, we use volume fraction instead. The volume fraction bounds are also known in some higher dimensions.}  of packed balls in a bounded and convex domain in $\mathbb{R}^2$ with at least two packed balls is bounded from above by ${\pi}/{\sqrt{12}}$ \cite{boroczky}. The bound is achieved  by hexagonal triangular lattice.

Now, for $r<H$, at most one point of $\Phi$ falls in $b(o,r)$. So $N(r)\leq 1,~r<H$.
For $r\geq H$, at least two balls can be packed in the densest packing of $b(o,r+H)$. So we can apply the upper bound on the volume fraction to obtain that
$
     \pi H^2 N(r)\leq {\pi^2(r+H)^2}/{\sqrt{12}},
$
for $r\geq H$. Combining both, we have
\begin{equation}
   N(r)\leq
   \begin{cases} 1, & r< H\\\frac{\pi}{\sqrt{12}}(1+\frac{2r}{H}+\frac{r^2}{H^2}),& r\geq H.
   \end{cases}
\end{equation}
For all $r>0$, $1+{2\pi r}/{(\sqrt{12}H)}+{\pi r^2}/(\sqrt{12}{H^2})$ is a $\mathbb{P}\mhyphen \mathrm{a.s.}$, upper bound for $N(r)$. 

 \subsection{Proof for Lemma \ref{l2}}
\label{appendix: l2}
First, we show that $L_{l,r} \triangleq\{y\in\mathbb{R}^2\colon \Phi(b(y,r)) \le l\}$
is a random closed set \cite[Chap. 9]{baccelli:hal-02460214}. 
We have
$\Phi(b(y,r)) =\sum_{X_i\in \Phi} \ind{\big(||X_i-y||<r\big)}.$
For all $i$, the function $y\to \ind{\big(||X_i-y||< r\big)}$ is lower semi-continuous because
it is the indicator function of an open set. 
This in turn implies that the function
$y\to \Phi(b(y,r))$ is lower semi continuous as well because only finitely many terms
show up in the sum in the neighborhood of any point.
Hence the lower level sets of $\Phi(b(y,r))$ are closed sets of $\mathbb{R}^2$. In order to prove that this closed set is a random closed set,
it is enough to show that for all compact sets $K$, the set
$\{\omega \in \Omega \colon L_{l,r}\cap K =\emptyset\}$
belongs to $\cal A$. To prove this, we can rewrite the last set as
$$\sum_{X_{i_1},\ldots,X_{i_{l+1}}\in \Phi, \ne} \prod_{j=1}^{l+1} \ind\big(\|X_{i_j}-K\|< r\big)>0,$$
where the last sum bears on all $l+1$-tuples of distinct points of $\Phi$ and $\|X_{i_j}-K\|$ denotes the minimum distance between $X_{i_j}$ and the set $K$.
Since there exists an enumeration of the points of $\Phi$ 
in terms of a countable collection of points, each of which is a random variable (e.g. based
on their distance to the origin),
the last set belongs to $\cal A$.

Let $A^c$ denote the complement of $A$ in $\mathbb{R}^2$. We have
$L_{l,r}^c=\{y\in \mathbb{R}^2 \colon \Phi(b(y,r)) > l\}.$
The fact that $L_{l,r}$ is a random closed set implies that $\overline{L_{l,r}^c}$,
where $\overline A$ denotes the closure of $A$, is a random closed set. Since
$\bigcap_{y\in\mathbb{R}^2}\{\Phi(b(y,r))\leq l\} =\{L_{r,l}^c =\emptyset\} 
=\{\overline{L_{r,l}^c} =\emptyset\}$,
and since one can rewrite
$\{\overline{L_{l,r}^c}=\emptyset\}$ as $\{\overline{L_{l,r}^c}\in {\cal F}^{\mathbb{R}^2}\}$, we 
get that 
$\bigcap_{y\in\mathbb{R}^2}\{\Phi(b(y,r))\leq l\}$
is an event for all $r$ and $l$. This proves the first part of the statement.

Now we prove the second part of the statement. Let $k_0=\lfloor g(0)\rfloor$.
The function $g$ is non-decreasing, with countable discontinuities, and
tends to $\infty$. Hence $\exists t_k,~k=0,1,...$, with $t_0=0$, and s.t. 
$\lfloor{g(r)\rfloor} = k +k_0$ for $r\in[t_k,t_{k+1})$. Now,
\begin{align*}
E_{\Phi} &=
\bigcap_{r>0} \bigcap_{y\in\mathbb{R}^2}
\{\Phi(b(y,r))\leq g(r)\} \\
&=
\bigcap_{k\geq 0}\bigcap_{r\in[t_k,t_{k+1})}\bigcap_{y\in\mathbb{R}^2} \{\Phi(b(y,r))\leq k+k_0\}\\
&= \bigcap_{k\geq0} D_k,
\end{align*}
with $ D_k\triangleq \bigcap_{y\in\mathbb{R}^2} \{\Phi(b(y,t_{k+1}))\leq k+k_0\}$.
Hence, to prove that $E_\Phi$ is an event, it is enough to prove that
for all $k$, $D_k$ is an event of $\cal A$. This follows from the first part of the statement.

 \label{appendix: thm-measure-ball-reg}
Lastly, it suffices to show that
$\mathbb{P} \big(\Phi(b(o,r))\le g(r)\big)=1$, for all $r>0$, implies that $\mathbb{P}(D_k)=1$ for all $k$.
As explained above, $\mathbb{P} \big(\Phi(b(o,r))\le g(r)\big)=1$, for all $r$, implies that
$\mathbb{P} \big(\Phi(b(o,t_{k+1}))\le k_0+k\big)=1$ for all $k$.
By stationarity, this implies that $\mathbb{P} \big(\Phi(b(y,t_{k+1}))\le k_0+k\big)=1$ for all $k$ and $y$.
Hence, for all $k$,
$\mathbb{P} \big(\bigcap_{y\in \mathbb{Q}^2}\Phi(b(y,t_{k+1}))\le k_0+k \big)=1$.
That is
$\mathbb{Q}^2 \subset L_{k_0+k,t_{k+1}}$ $\mathbb P$-a.s.
But since $L_{k_0+k,t_{k+1}}$ is a closed set, then 
$\mathbb{R}^2 \subset L_{k_0+k,t_{k+1}}$ $\mathbb P$-a.s.
So $\mathbb{P}(D_k)=1$.
This concludes the proof.

  \subsection{Proof for Lemma \ref{thm: measure-void-reg}}
  \label{appendix: thm-measure-void-reg}
 Recall that the definition for void regulation involves closed balls.
For reasons dual to those explained in Appendix \ref{appendix: thm-measure-ball-reg}, the upper-level set 
$V_\Phi=\bigcap_{y \in \mathbb{R}^2}\{ \Phi(B((y,r))\ge 1\}$
is a random closed set, which in turn implies that $V_\Phi=\mathbb{R}^2$ is an event.
If, in addition, $\mathbb{P}\big(\Phi(B(o,r))\ge 1\big)=1$, then
$\mathbb{P}\big(\bigcap_{y \in \mathbb{Q}^2} \{\Phi(B(y,r))\ge 1\}\big))=1$.
But since $V_\Phi$ is a closed set, the last property implies that
$\mathbb{P}(V_\Phi)=1$. This completes the proof of the theorem.
 

\subsection{Proof for Theorem \ref{thm: csp-lb}}\label{appendix: thm-bipolar-P_m}
\vspace{-2em}
\allowdisplaybreaks
\begin{align*}
&\Pr_\Psi^o( \mathrm{SINR}>\theta\mid\Phi) \\   &=
\Pr_\Psi^o\left( I+W < h_t\ell(\tau)/\theta\mid\Phi\right)\\
& \peq{a} \Pr_\Psi^o\left(I < h_t\ell(\tau)/\theta-W,h_t>W\theta/\ell(\tau)\mid \Phi\right)
\\       &= \int_{\frac{W\theta}{\ell(\tau)}}^{\infty} f_{h_t}(x) \Pr_\Psi^o\left(I< x\ell(\tau)/\theta-W\mid\Phi,h_t=x\right) \dd x\\
       & \peq{b} \int_{\frac{W\theta}{\ell(\tau)}}^{\infty} f_{h}(x) \left(1-\Pr_\Psi^o\left(\exp(sI)\geq \exp\left(s\left(\frac{x\ell(\tau)}{\theta}-W\right)\right)~\bigg|~\Phi\right)\right) \dd x
       \\
        & \pgeq{c} \int_{\frac{W\theta}{\ell(\tau)}}^{\infty} f_{h}(x)\left(1-\inf_{s\in[0,s^{*})}\mathcal{L}_{I\mid\Phi}(s)\exp\left(-s\left(x\ell(\tau)/\theta-W\right)\right)\right)^+\dd x\\ 
        &\pgeq{d} \int_{\frac{W\theta}{\ell(\tau)}}^{\infty} f_{h}(x)\left(1-e^{\inf_{s\in[0,s^*)}A_{\Tilde\ell}-\tilde\ell(\tau)-{s\left(x\ell(\tau)-\theta W\right)}/{\theta}}\right)^+\dd x
\end{align*}
where $\Tilde\ell(r) = \log\mathcal{L}_h(-s\ell(r))$, and $f^+\triangleq\max{(0,f)}$.
Step (a) follows from the fact that the interference is always non-negative.
Step (b) follows from the independence of the interference and the fading from the desired transmitter. Step (c) and (d) use the Chernoff bound as in Corollary \ref{corr: Chernoff}. 
The second part can be easily seen as $\zeta(\theta)$ is continuous, non-increasing, and $\lim_{\theta\to0} \zeta(\theta)=1$.

\section*{Acknowledgements}
This work was supported by the ERC NEMO grant, under the European Union's Horizon 2020 research
and innovation programme, grant agreement number 788851 to INRIA.
The authors thank A. Khezeli for his comments and suggestions on this work.

\bibliographystyle{IEEEtran}

\bibliography{ref,bookref,ref2}

\end{document}